\documentclass[11pt]{article}
\usepackage{times,epsfig,epsf,psfrag,array,amsmath,amssymb,verbatim,mathtools}
\usepackage{url,comment,enumerate,graphicx,graphics,wrapfig}
\usepackage[usenames,dvipsnames,svgnames,table]{xcolor}
\usepackage{units}
\usepackage{amsthm}
\usepackage{bm}
\usepackage{graphicx}
\usepackage{caption}
\usepackage{wrapfig}
\usepackage{subcaption}
\usepackage{hyperref}

\makeatletter
\def\thm@space@setup{%
  \thm@preskip=\parskip \thm@postskip=0pt
}
\makeatother

\newtheorem{theorem}{Theorem}[section]
\newtheorem{lemma}[theorem]{Lemma}
\newtheorem{claim}[theorem]{Claim}
\newtheorem{corollary}[theorem]{Corollary}
\newtheorem{observation}[theorem]{Observation}

\newtheorem*{remark}{Remark}
\long\def\comment #1\commentend{}

\def\Prb{\mbox{\textup{Pr}}}

\def\a{\alpha}
\def\b{\beta}
\def\r{\lambda}
\def\X{\bm{x}}
\def\Xnr{\X^{n,\r}}
\def\discost{\psi}
\def\F{\mathcal{F}}
\def\H{\mathcal{H}}
\def\M{\mathcal{M}}
\def\cL{\mathcal{L}}

\def\Cfree{C^{free}}

\def\Ebb{\mathbb{E}}

\def\fthotel{\mbox{\texttt{FPH}}}
\def\PoA{\textsf{PoA}}
\def\PoS{\textsf{PoS}}
\def\NE{\Xnr}
\def\FF{\bm{y}^{n,0}}
\def\OPT{\bm{o}^{n,\r}}

\newcommand{\Xm}[1]{\X_{-#1}}
\newcommand{\mean}[1]{\hat{#1}}
\newcommand{\EH}[1]{\mean{\H}\left(#1\right)}
\newcommand{\EM}[1]{\mean{\M}\left(#1\right)}
\newcommand\numberthis{\addtocounter{equation}{1}\tag{\theequation}}

\usepackage[english]{babel}

\title{Hotelling Games with Multiple Line Faults}

\author{
Avi Cohen\thanks{Weizmann Institute of Science, Rehovot, Israel.
\{avi.cohen,david.peleg\}@weizmann.ac.il}
\and
David Peleg$^*$
}

\begin{document}

\maketitle

\begin{abstract}
	The Hotelling game consists of $n$ servers each choosing a point on the line segment, so as to maximize the amount of clients it attracts. Clients are uniformly distributed along the line, and each client buys from the closest server. In this paper, we study a \emph{fault-prone} version of the Hotelling game, where the line fails at multiple random locations. Each failure disconnects the line, blocking the passage of clients. We show that the game admits a Nash equilibrium if and only if the rate of faults exceeds a certain threshold, and calculate that threshold approximately. Moreover, when a Nash equilibrium exists we show it is unique and construct it explicitly.

 Hence, somewhat surprisingly, the potential occurrence of failures has a stabilizing effect on the game (provided there are enough of them). Additionally, we study the social cost of the game (measured in terms of the total transportation cost of the clients), which also seems to benefit in a certain sense from the potential presence of failures.
\end{abstract}
\section{Introduction}

\paragraph{Background and Motivation.} The Hotelling game, introduced by Hotelling in the seminal~\cite{hotelling1929stability}, is a widely studied model of spatial competition. It considers two servers, each can choose where to set its shop along a street (a segment). Clients are assumed to be uniformly distributed along the street, and to shop at the closest server. As a consequence, the payoff of each server is equal to the size of the segment of points closer to it than to the other server (a.k.a. its \emph{Voronoi cell}).  A benevolent central planner, serving the public good (i.e., the interests of the clients), would require the two servers each to occupy one of the points $1/4$ and $3/4$. 
However, this does not produce an equilibrium, as each server can profit by moving closer to the other. In equilibrium, both servers occupy the center of the line. This is suboptimal for the clients, as they have to travel twice as far as in the social optimum, on average, to reach the closest server.
 Thus emerges what has later been called the ``Principle of Minimum Differentiation''~(\cite{boulding1966economic}), according to which competitors tend to cluster together. Hotelling stressed that this game expresses not only \emph{geographic} location, but also product selection (which can be viewed as location within a characteristic space), and a candidate's political platform (viewed as location on the political spectrum). Accordingly, the principle of minimum differentiation also expresses the tendency to reduce variety in products and policies. That is, competition between companies often diminishes the choices available to consumers rather than expanding them.
The 3-player version of the Hotelling game, studied later~(\cite{lerner1937some}), has no Nash equilibrium, since each server would seek to take up a position close to one of the other servers but not between them.

Game Theory typically assumes a reliable environment and rational payoff maximizing players. In reality, however, no system is infallible; players are prone to mistakes and unexpected behavior, and the environment itself is subject to failures. Therefore, it is important to examine how player strategies are affected by failures in components of the game.
In particular, in the context of the Hotelling game, it is natural to consider an extension of the game to an $n$-player {\em fault-prone} Hotelling game. In such a game, $n$ servers choose a location along the $[0,1]$ segment, and as their payoff they get the expected size of their Voronoi cell in the presence of random \emph{line faults}.
A line fault is a disconnection at a point on the line through which clients cannot pass.
Such line disconnections represent temporary road blockages, which may be caused by many different reasons, e.g., construction work, traffic accidents, congested intersections, parades, public demonstrations, etc. Such blockages clearly affect the profits of shops along the street, and therefore change their optimal location strategies.

We are interested in studying the effects line faults have on the game, including the existence and structure of Nash equilibria, along with other types of influences, such as the changes in total transportation cost compared with the social optimum (possibly taking into account also the possibility that clients may become ``disconnected'' from all servers when surrounded by faults.
Understanding this setting calls for the study of location games through the perspective of fault-tolerant distributed systems, and for combining ideas from the fields of Game Theory and Computer Science.

The $n$-player Hotelling game on the line with a \emph{single line fault} was recently introduced in Avin et al.~\cite{ACLP18}, and its Nash equilibria were completely characterized for every $n$.
However, the framework of~\cite{ACLP18}, which is based on a uniform distribution of the fault location over the line, and the direct analysis method employed therein, do not extend easily to settings allowing the possibility of \emph{multiple line faults} occurring simultaneously. In particular, in such a generalized setting, it is necessary to define the distribution of failures so as to make the analysis feasible.
In the current paper, we address this generalized fault setting, define a suitable model (based on Poisson distributed failures) and analyze the properties of Nash equilibria in it.

Our choice of the Poisson process as the distribution of faults may seem restrictive, but it turns out to be an excellent model for many different real-world phenomena, and is commonly used as the distribution of random and independent points in continuous space across many different disciplines, including, among others, queueing theory~\cite{kleinrock1975queuing}, wireless networks~\cite{gilbert1961random} and continuum percolation~\cite{meester1996continuum}.

Interestingly, the possibility of two simultaneous faults seems to radically change player behavior, when compared to the single fault variant. Notably, we show that the possibility of more than one fault occurring forces the servers to distance themselves from their neighbors. Hence in this setting, the principle of minimum differentiation no longer holds, and client transportation cost in is reduced compared to the single fault model. Another difference concerns the existence of Nash equilibria. For example, the 3-player game admits an equilibrium in a model allowing multiple (randomly distributed) failures, but not in the failure-free or single-failure models.

\paragraph{Contributions.} In our model the faults are distributed along the line according to a Poisson process with parameter $\r$. Informally, each point on the line may be faulty with some small probability, such that the number of faults is $\r$ in expectation.
We start by characterizing the Nash equilibria of the 2-player game (Theorem~\ref{thm:NE-2players}). For $n\geq3$, we show that every equilibrium is what we call a \emph{canonical profile} (Theorem~\ref{thm:NE-form}), i.e., the distance between every pair of neighboring players is constant, and the distance from the leftmost player to 0 (a.k.a. the \emph{left hinterland}) is the same as the distance from the rightmost player to 1 (a.k.a. the \emph{right hinterland}). Moreover, the canonical profile, when it exists, is uniquely defined for every $n$ and $\r$. This reduces our problem to considering whether the canonical profile is a Nash equilibrium for given values of $n$ and $\r$. As our main result, we show that a Nash equilibrium exists for the $n$-player game if and only if $\r\geq \r_{\min}(n)$, for some \emph{threshold function}  $\r_{\min}(n)$ (Theorem~\ref{thm:NE-lower-bound}). Additionally, we show a way to efficiently approximate the values of $\r_{\min}(n)$ to any precision.

Finally, we also studied the effects of failures on social cost, measured
in two different ways. First, we considered failure free executions,
and looked at the total transportation cost of the clients. We compared
this cost for the Nash equilibria of our game (in failure-free executions)
against the cost incurred on the clients in the Nash equilibria obtained in the classical (failure-free) Hotelling game (see Sect. \ref{sec:pos}).
The limitation of this type of comparison is that it deals only with the
behavior of strategies in failure free executions. Our second approach was
therefore to concentrate on executions with failures.
In such executions it is possible for clients to get completely disconnected.
In such a case, we set the cost for a disconnected client to some (large)
constant. We similarly compared the total cost to clients in Nash equilibria of our game and in Nash equilibria of the classical game, both under the possibility of faults.

\paragraph{Related Work.}%
Coping with faults constitutes a fertile area of research in Computer Science. Specifically, there is a large body of research on fault tolerant facility location problems~(\cite{chechik2014robust,chechik2015fault,khuller2000fault,snyder2016or,sviridenko2002improved,swamy2008fault}). Recently, the issue of failures and their effects has been considered in the field of Game Theory~(\cite{eliaz2002fault,kalai2004large,feige2011mechanism,GR14,meir2012congestion}). Nevertheless, relatively little work has been done to apply game theory to fault tolerant facility location. The few papers that do consider this problem, such as Wang and Ouyang~\cite{wang2013continuum} and Zhang et al.~\cite{zhang2016competitive}, consider location models that are far removed from the Hotelling game we consider here.

Hotelling's model and its many variants have been studied extensively.
Eaton and Lipsey~\cite{eaton1975principle} extended Hotelling's analysis to any number of players and different location spaces. Our model is a direct extension of their $n$-player game on the line segment. d'Aspremont et al.~\cite{d1979hotelling} criticized Hotelling's finding and showed that when players compete on price as well as location, they tend to create distance from one another, otherwise price competition would drop their profit to zero. Our results show differentiation in location in the classic $n$-player Hotelling game without introducing price competition. A large portion of the Hotelling game literature is dedicated to models with price competition. We, however, exclusively consider pure location competition models since they are closely related to the facility location problem studied in Computer Science. Eiselt, Laporte and Thisse~\cite{eiselt1993competitive} provide an extensive comparison of the different models classified by the following characteristics: the number of players, the location space (e.g., circle, plane, network), pricing policy, the behavior of players, and the behavior of clients. (For more recent surveys see Eiselt et al.~\cite{eiselt2011equilibria} and Ashtiani~\cite{ashtiani2016competitive}.)

Randomness in client behavior was introduced by De Palma et al.~\cite{de1987existence}. Their model assumes client behavior has an unpredictable component due to unquantifiable factors of personal taste, and thus clients have a small probability of ``skipping'' the closest server and buying from another. In their model, all servers would locate at the center in equilibrium, reasserting Hotelling's conclusion. In our model, clients exhibit randomness in their choice of servers as well, but in equilibrium servers create a fixed distance from their neighbors. In a similar more recent work, Peters et al.~\cite{peters2015waiting} modified client behavior by introducing queues to each server. Clients would thus behave strategically to minimize the distance and the waiting time as a total cost.

Several recent works have considered Hotelling games on networks (\cite{palvolgyi2011hotelling,fournier2016general,fournier2016hotelling}). P\'{a}lv\"{o}lgyi's~\cite{palvolgyi2011hotelling} main result was that on any network there exists a Nash equilibrium when there are sufficiently many players in the game. While we have not yet considered networks as an action space, our results imply that for any number of players there exists a Nash equilibrium when there are sufficiently many faults. Fournier~\cite{fournier2016general} considers a general distribution of clients along the network, but assumes clients always choose the closest server, i.e., the Voronoi cells are as in the original model. We, on the other hand, assume the clients are uniformly distributed but the Voronoi cells are modified by faults. The discrete version of Hotelling games on networks (a.k.a. Voronoi games) were considered by Marvonicolas et al.~\cite{mavronicolas2008voronoi} and Feldman et al.~\cite{feldmann2009nash}.

Avin et al.~\cite{ACLP18} introduced the fault prone Hotelling game. They studied two types of faults: line faults and player faults. The Hotelling game with line faults they considered was limited to a single fault, and the resulting Nash equilibria (of which there are none for $n=3$, exactly one for $n=1,2,4,5$, and infinitely many for $n\ge 6$) were similar to the ones obtained in the failure-free setting, except with the players located slightly closer to the center. In this paper, we show that when multiple faults are considered, a more interesting picture emerges, and the resulting equilibria are significantly different than those obtained in the failure-free setting. The other fault model they considered was the Hotelling game with player faults, where each player has some probability of being removed from the game. It is shown there that in the player faults setting there exists no Nash equilibrium for $n\geq3$ players.


\section{The Model}
\paragraph{The Game.}
In this section we present the location model studied in this paper.
\comment
\paragraph{Servers and Clients.}
\commentend
The system consists of a finite set $N=\{1,\ldots,n\}$ of servers (acting as the \emph{players} in our game formulation), each of whom has to decide where to set shop along the interval $[0,1]$. We assume that clients are uniformly distributed along the line, and that they choose the closest server that is not disconnected from them. Each server wants to maximize its expected market share in the presence of faults. It is possible for more than one server to occupy the same location. In that case, clients choosing that location are divided equally between the servers at that location.

\comment
\paragraph{Faults.}
\commentend
The model assumes that faults, or disconnections, occur at random along the line segment $[0,1]$, and clients cannot choose servers beyond a disconnection. The set of faults, denoted as $\F$, is distributed along the line at random according to a Poisson process with rate $\r>0$. More specifically, the probability that $k$ faults occur in any segment $[a,b] \subseteq [0,1]$ is
$$
    \Prb\left[|\F \cap [a,b]|=k\right]= \frac{e^{-\r(b-a)}\cdot (\r(b-a))^k}{k!}~.
$$
Namely, it is a Poisson distribution with rate $\r(b-a)$. Intuitively, the Poisson process could be viewed as a continuous analogue of the Binomial distribution. Divide the $[0,1]$ line into small segments of length $\delta>0$ and let $p=\frac{\r}{\delta}$. At each $\delta$ segment a fault occurs with probability $p \cdot \r$. As $\delta \to 0$ the number of faults in a segment $[a,b] \subseteq [0,1]$ converges to the Poisson distribution with rate $\r(b-a)$. Alternatively, we can define the Poisson process as a sequence of exponential random variables, i.e., the distance between each pair of consecutive faults is an exponential random variable with rate $\r$. (For a formal definition of the Poisson process see Appendix~\ref{sec:poisson-process}.)

\comment
\paragraph{Markets in the presence of faults.}
\commentend
Consider a specific instance of the system. In this instance, let $\X = (x_1,\ldots,x_n) \subset [0,1]$ be the profile of the server locations. Let $ \F = \{f_1,\ldots, f_k\} $ be a given set of faults that have occurred. Two servers $i,j\in N$  are said to be \emph{colocated} if $x_i=x_j$. For $i\in N$, the set of $i$'s colocated servers is defined as $ \Gamma_i= \{j \in N \mid x_j=x_i\}$, and the size of this set is defined as $\gamma_i= |\Gamma_i|$. A server that is not colocated with other servers is called \emph{isolated}. Two servers are called \emph{neighbors} if no server is located strictly between them. A \emph{left (resp., right) peripheral server} is a server that has no servers to its left (resp., right). The servers divide the line into \emph{regions} of two types: \emph{internal regions}, which are regions between two neighbors, and two \emph{hinterlands}, which include the region between 0 and the left peripheral server, and the region between the right peripheral server and 1. (See Figure~\ref{fig:NE-shape}, where the two hinterlands are marked by $H$.)

The \emph{market} of each server $i\in N$ is the line segment in which clients choose location $x_i$. Note that colocated servers have the same market. Let $x_i^\ell$ and $x_i^r$ be the locations of $i$'s left and right neighbor, respectively. When no left (resp., right) neighbor exists we define $x_i^\ell=-1$ (resp., $x_i^r=2$). Let $f_i^\ell$ and $f_i^r$ be the closest faults to the left and right of $i$, respectively. When there are no faults to the left (resp., right) of $x_i$ we define $f_i^\ell=-\infty$ (resp., $f_i^r=\infty$). We define the market of $i\in N$ as the segment $[L_i,R_i]$, where $L_i$ and $R_i$ are defined as follows:
$$
\hbox{$
L_i=\left\{
				\begin{array}{ll}
					f_i^\ell,		& \mbox{if } f_i^\ell \geq x_i^\ell~; \\
                    0,               & \mbox{if } f_i^\ell < x_i^\ell=-1~; \\
					x_i^\ell, 	 & \mbox{otherwise.}
				\end{array}
			\right.
$}
~~~~~~~~~~~~~~
\hbox{$
	R_i=\left\{
				\begin{array}{ll}
					f_i^r,		& \mbox{if } f_i^r \leq x_i^r~; \\
                    1,          & \mbox{if } f_i^r > x_i^r=2~; \\
					x_i^r, 	 & \mbox{otherwise.}
				\end{array}
			\right.
$}
$$
In the definition of $L_i$,
the first case handles a situation where a failure occurs to the left of $i$ but to the right of its left neighbor if exists. The second concerns the case where there are neither neighbors nor failures to the left of $i$. The third handles a case where there is no failure between $i$ and its left neighbor.
The definition of $R_i$ is analogous.


Server $i$'s market is divided into two parts, $[L_i,x_i]$ and $[x_i,R_i]$, referred to as server $i$'s \emph{left} and \emph{right} half-markets, respectively.

\comment
\paragraph{The Game.}
\commentend
The \emph{fault-prone Hotelling game} is denoted as $\fthotel(n,\F)$, where $N=\{1,\ldots,n\}$ is the set of players and $\F$ is the distribution of faults. For $i \in N$, the \emph{action} of player $i$ consists of selecting its location, $x_i \in [0,1]$ . The vector $ \X = (x_i)_{i\in N}$ is the profile of actions. Let $\Xm i$ denote the profile of actions of all the players different from $i$. Slightly abusing notation, we will denote by $(x'_i, \Xm i)$ the profile obtained from a profile $\X$ by replacing its $i$th coordinate $x_i$ with $x'_i$.

We denote the \emph{size} of server $i$'s market $[L_i,R_i]$ by $D_i = R_i- L_i$. This is itself a random variable, and we define the \emph{payoff} of player $i$ given the profile $\X$ as the expectation (over the possible failure configurations) of $D_i$ divided by the number of players colocated with $i$, namely,
$$u_i(\X)= \frac{\Ebb\left[ D_i \right]}{\gamma_i}~.$$
However, for the analysis, it is more convenient to view the payoff as composed of the \emph{left} and \emph{right payoffs}, $D_i=D_i^\ell+D_i^r$, where $D_i^\ell=x_i-L_i$ and $D_i^r=R_i-x_i$, and analyze $D_i^\ell$ and $D_i^r$ separately. The reason for this is that, as it turns out, $\Ebb[D_i^r]$ is only a function of the length $x_i^r-x_i$ of the right region of $i$ and the failure distribution, and similarly $\Ebb[D_i^\ell]$ depends only on $x_i-x_i^\ell$ and the failures, making it easier to analyze them separately.

Given a profile $\X$, $x'_i\in[0,1]$ is an \emph{improving move} for player $i$ if $u_i(x'_i,\Xm i) > u_i(\X)$. $x^*_i\in X$ is a \emph{best response} for player $i$ if $u_i(x^*_i,\Xm i) \geq u_i(x'_i,\Xm i)$ for every $x'_i\in [0,1]$.

A profile $\X^*$ is a \emph{Nash equilibrium} if no player has an improving
move, i.e., for every $i\in N$ and every $x_i\in X$,
$u_i(\X^*) \geq u_i(x_i,\Xm i^*)$.

\paragraph{Efficiency of Equilibria.}
The Hotelling game with no faults, $\fthotel(n,0)$, is a constant-sum game, making the sum of player payoffs a poor measure of efficiency since it does not depend on player strategies. Therefore, Fournier and Scarsini~\cite{fournier2016hotelling} measure the efficiency of Nash equilibria in terms of the client's transportation costs. While our fault-prone Hotelling game is not a constant-sum game, we recognize the importance of the clients' transportation costs as a measure of efficiency, and consider it as a basis for comparison to variants of the Hotelling game other than our own.

Following Fournier and Scarsini~\cite{fournier2016hotelling},
for $\X\in[0,1]^n$ and $ y\in[0,1] $ define the transportation cost for a client
located in $y$ to be the distance from $y$ to the closest server,
$$ d(\X,y) = \min_{i\in\{1,\ldots,n\}}|y-x_i|~.$$
The \emph{failure-free transportation cost} $\Cfree(\X)$ is
$$ \Cfree(\X)=\int_0^1d(\X,y)dy~,$$
namely, the total transportation cost for clients, when each one of them buys
from the closest server. Notice that this definition of $\Cfree(\X)$ assumes that
no failures have occured, and depends on the profile of player locations alone.
The definition is not easily amenable to adaptation to scenarios with failures,
since once faults actually occur, each client has some probability of being
disconnected from all players, making its expected transportation cost infinite.

To overcome this difficulty, we explore two different approaches.
The first
is based on using the measure $\Cfree$ but restricting
the class of scenarios on which we compare the cost of profiles.
While our model accommodates failures, and the player strategies must take
this possibility into account, we measure and compare transportation cost
exclusively in the case where no faults have occurred. In other words,
we examine how the \emph{possibility} of failures affects player strategies
and, as a direct result, the resulting total cost for the clients.

Consider a game $\fthotel(n,\r)$ that has a Nash equilibrium. We denote by
$S_{n,\r}$ the set of pure Nash equilibria of the game $\fthotel(n,\r)$ and
define the \emph{price of anarchy} and the \emph{price of stability},
respectively, as
$$ \PoA(n,\r) = \frac{\sup_{\X\in S_{n,\r}}\Cfree(\X)}{\inf_{\X\in [0,1]^n}\Cfree(\X)}~,
~~~~~~~~~~
 \PoS(n,\r) = \frac{\inf_{\X\in S_{n,\r}}\Cfree(\X)}{\inf_{\X\in [0,1]^n}\Cfree(\X)}~. $$

Our second approach to measuring client costs is to allow also scenarios with failures, and associate a large but
finite cost $\discost\gg 1$ with client disconnection.
Hence for $y\in[0,1]$ we define the cost
for a client located in $y$ under a given set of faults $\F$
(the set of disconnection points) as follows.
A server $x$ is said to be {\em accessible} from $y$ under $\F$
if the line segment between $x$ and $y$ contains no point from $\F$.
For $\X\in[0,1]^n$, denote the set of servers accessible from $y$ under $\F$
by $\X_\F(y)$. If $\X_\F(y)=\emptyset$ then $y$ is {\em disconnected}.
Define the {\em access cost} from $y$ to $\X$ under $\F$ as
$$d_\F(\X,y) ~=~
\begin{cases}
  d(\X_\F(y),y), & \X_\F(y) \ne \emptyset~, \\
  \discost, & y~\mbox{ is disconnected}~.
\end{cases}$$
The \emph{access cost under fault set $\F$}, $C_\F(\X)$, is
$$C_\F(\X)=\int_0^1 d_\F(\X,y)dy~.$$
Letting $Y^c_\F(\X)$ and $Y^{dc}_\F(\X)$ denote the set of connected
(resp., disconnected) clients of $\X$ under failure set $\F$,
the access cost can be separated into
$$C_\F(\X) = C_\F^{c}(\X) + C_\F^{dc}(\X),$$
where $C_\F^{c}(\X)$ is the {\em transportation cost} of the connected clients
$Y^c_\F(\X)$, and $C_\F^{dc}(\X)$ is the {\em disconnection cost} incurred by
the disconnected clients $Y^{dc}_\F(\X)$.

\section{Main Results}
\label{sec:results}

\paragraph{Characterization of Stable Profiles.}
Let us first characterize the profiles that lead to a Nash equilibrium for the fault-prone Hotelling game.
Given an integer $n\geq 2$ and a real $\r > 0$, the pair $(H,M)$, $H,M\in[0,1]$, is called a \emph{canonical pair} if $H$ and $M$ satisfy the following equations:
\begin{align}
\label{eq:optimal-H}
&  e^{\lambda(M-H)} = \frac{1+\lambda M}2
\\
\label{eq:total-regions}
&  2H+(n-1)M = 1
\\[1ex]
\label{eq:minimal-lambda}
&  \r H > \ln2
\end{align}
\noindent The canonical pair induces a profile $\Xnr$ for the game $\fthotel(n,\r)$, such that
$$ x^{n,\r}_i= H+(i-1)M$$
for every $i \in N$ (see Figure~\ref{fig:NE-shape}). We refer to this profile as the {\em canonical profile}.

\begin{figure}[ht]
\includegraphics[width=\textwidth]{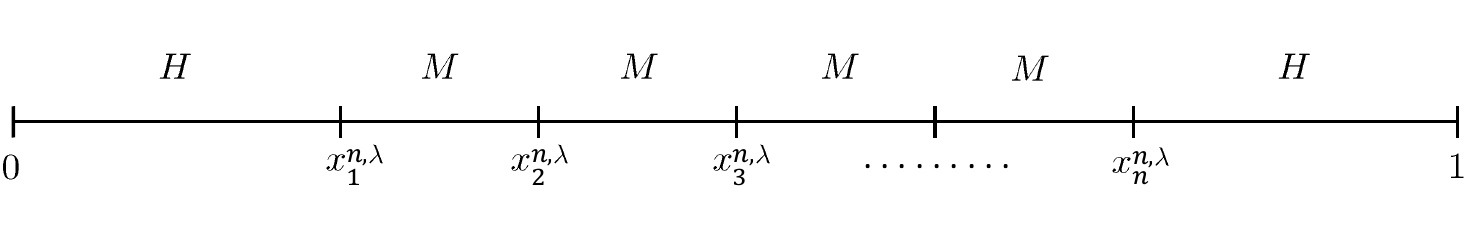}
\caption{The canonical profile.}
\label{fig:NE-shape}
\end{figure}

\begin{remark}
Whenever it exists, the canonical pair is uniquely defined for $n$ and $\r$. However, we have no closed formula for $H$ and $M$ as a function of $n$ and $\r$. In spite of this fact, our analysis does yield a parametric representation for $M$ as a function of $\r$. Namely,
$$(\r(t),M(t)) ~=~
            \left(\,(n+1)t-2\ln\left((1+t)/2\right)~
                        ,~1/\left({n+1-2\ln\left((1+t)/2\right)}/t\right)\,\right)~.$$
\end{remark}

\paragraph{Existence of Nash Equilibria.}
Our first result is that the fault-prone Hotelling game for two players always has a unique Nash equilibrium. Moreover, if a canonical profile exists, then it is the unique Nash equilibrium.

\begin{theorem}\label{thm:NE-2players}
    For $n=2$ players, the game $\fthotel(2,\r)$ has a unique Nash equilibrium, which is given by
    $$ \X^*=\left\{
				\begin{array}{ll}
					\;\;\Xnr\;,		         & \mbox{if } \r>2\ln2~; \\
                    \left(\frac12,\frac12\right),    & \mbox{otherwise.}
				\end{array}
			\right. $$
\end{theorem}
\noindent
For $n\ge 3$
players, we show that the only possible Nash equilibrium is the canonical profile.

\begin{theorem}\label{thm:NE-form}
  For $n\geq 3$ players, if the game $\fthotel(n,\r)$ has a Nash equilibrium,
  then it is unique and equal to the canonical profile $\Xnr$.
\end{theorem}

Additionally, we have found sufficient and necessary conditions over the canonical pair $H$ and $M$ for the existence of a Nash equilibrium (see Section~\ref{sec:NE-existence}). For any given $n$ and $\r$, we can derive $H$ and $M$, and determine whether the canonical profile is a Nash equilibrium using those conditions. However, our goal is to find a simple description for the set of values of $\r$ and $n$ for which a Nash equilibrium exists. In fact, our main theorem shows
that there exists a \emph{threshold function} $\r_{\min}(n)$ such that the game $\fthotel(n,\r)$ admits a Nash equilibrium if and only if $\r \geq \r_{\min}(n)$. Moreover, we show there exists a global constant $\a_0$ from which follows an \emph{exact} formulation of the threshold function $\r_{\min}(n)$. While we do not have an exact value of $\a_0$, we know that it is given by the intersection of two implicit functions (see Equations~\eqref{eq:alpha-beta3} and \eqref{eq:alpha-beta4}), which may be approximated to arbitrary precision. Moreover, any upper (resp., lower) bound for $\a_0$ translates to an upper (resp., lower) bound for $\r_{\min}(n)$.

\begin{theorem}\label{thm:NE-lower-bound}
There exists a constant $\a_0>0$ such that the game $\fthotel(n,\r)$, for $n\geq3$ players,
admits a Nash equilibrium if and only if
        $$\r ~\geq~ \r_{\min}(n) ~=~ (n+1)\a_0-2\ln\left(\frac{1+\a_0}2\right)~.$$
Moreover, $\a_0 \approx 0.58813$, implying
that a Nash equilibrium exists if and only if
$$\r ~\geq~ \r_{\min}(n) ~\approx~ 0.58813n+1.04931~.$$
\end{theorem}

We conclude with the following corollary.
\begin{corollary}\label{cor:NE-profiles}
There exist global constants $c_{\min}$ and $c_{\max}$ such that if the game $\fthotel(n,\r)$ has a canonical pair $(H,M)$ then
	$$ H/M \geq 1-c_{\max}~, $$
and if $n\geq 3$ and the game admits a Nash equilibrium, then it also holds that
	$$ H/M \leq 1-c_{\min}~. $$
Moreover, $c_{\max} \approx 0.232$ and $c_{\min} \approx -0.392$.
\end{corollary}




\section{Overview of the Analysis}
This section provides a brief summary of our analysis, which appears in full detail in Appendix~\ref{sec:analysis}. We proceed in three main steps: calculating the \emph{payoff functions} of players under faults, optimizing each player's payoff \emph{locally} (within the interval between its neighbors), and finally verifying that no player has an improving move \emph{globally} (which implies a Nash equilibrium).

\paragraph{Payoff Functions.}We start by calculating the players' payoffs (Section~\ref{sec:region-profit}),
considering the left and right regions of each player separately. If a fault occurred in a region then the payoff is simply the distance to the closest fault. Otherwise, the payoff is half the length of the region for an internal region and the entire length of the region for a hinterland region.
Denote the expected profit from a hinterland (respectively, an internal region) of length $d$ by $\EH{d}$ (resp., $\EM{d}$).
Using the law of total probability we obtain that
$\EH{d}=\left[1-e^{-\r d}\right]/\r$ and
$\EM{d}=\left[1-e^{-\r d}\left(1+\r d/2\right)\right]/\r$.
This gives a simple function for the payoff $u_i(\X)$ of player $i$ over any profile $\X$, as the sum of the profits from the two adjacent regions.
\paragraph{Local Optimization.} The next step towards finding the Nash equilibria of a game $\fthotel(n,\r)$ is to optimize player location locally (Section~\ref{sec:local-optimum}), i.e., maximize $u_i(\,\cdot\,,\Xm{i})$ within the interval $[x_i^{\ell},x_i^{r}]$ to ensure no player has an improving move within the interval between its neighbors.
As the function $u_i(t,\Xm{i})$ is continuous and concave for
$t\in[x_i^{\ell},x_i^{r}]$, finding the maximum is reduced to solving the equation $\partial \, u_i(t,\Xm{i})/ \partial t = 0$.

For an internal player $i$, the optimal location within $[x_i^{\ell},x_i^{r}]$ is the center of the interval, $(x_i^{r}-x_i^{\ell})/2$. It follows that in every Nash equilibrium profile all internal regions are equal. For a peripheral player $i$, the optimal location between its neighbor and the border of the line is given by an implicit function (Equation~\eqref{eq:optimal-H}), which has no closed-form solution for general $n$ and $\r$. For $n=2$, the locally optimal location is also globally optimal, yielding Theorem~\ref{thm:NE-2players}.
For $n\geq3$, these two results show that if a game $\fthotel(n,\r)$ admits a Nash equilibrium, then it must be the canonical profile $\Xnr$, implying Theorem~\ref{thm:NE-form}.


\paragraph{Global Stability.} To check the existence of a Nash equilibrium for a game $\fthotel(n,\r)$, it is left find the canonical profile $\Xnr$ by solving Equations~\eqref{eq:optimal-H}, \eqref{eq:total-regions} and \eqref{eq:minimal-lambda}, and check whether any player has an improving move. Actually, we show that it suffices to only check whether an internal player can improve its payoff by moving to the hinterland (as in Figure~\ref{fig:internal-improve}). Hence, the following condition is sufficient and necessary for $\Xnr$ to be a Nash equilibrium (Observation~\ref{obs:NE-conditions}):
$$ \EM{M}+\EM{M} \geq \max_{t\in[0,H]}\left(\EH{t}+\EM{H-t}\right)~.$$
However, there is no direct method to find an exact solution for Equations~\eqref{eq:optimal-H}, \eqref{eq:total-regions} and \eqref{eq:minimal-lambda} for general $n$ and $\r$. To circumvent this difficulty, we use a reparameterization of the variables $n$, $\r$, $H$ and $M$ that simplifies the analysis. Specifically, we define $\a=M\r$ and $c=1-H/M$ and show that the values of $c$, $M$, $H$ and $\r$ can be expressed in terms of $\a$ and $n$ as follows:
$c=\ln\left(\frac{1+\a}{2}\right) / \a$;
$M=\frac{1}{n+1-2c}$;
$H=\frac{1-c}{n+1-2c}$; and
$\r=\a(n+1)-2\ln\left(\frac{1+\a}{2}\right)$.
As a consequence, for every $n$ and $\a>0$ we may find $\r(\a)$, $M(\a)$ and $H(\a)$ such that $(H(\a),M(\a))$ is the canonical pair of the game $\fthotel(n,\r(\a))$. Plugging the above into Equations~\eqref{eq:optimal-H}, \eqref{eq:total-regions} and \eqref{eq:minimal-lambda} and the condition we obtained for a Nash equilibrium (Observation~\ref{obs:NE-conditions}), we find that for $\a>0$ the game $\fthotel(n,\r(\a))$ has a Nash equilibrium if and only if there exist $\b_1 \geq 0$ and $\b_2 \geq 0$ such that
\begin{align*}
	e^{-\a}(1+\a) &=e^{-2\b_1}(1+\b_1)~,  \\
	e^{-\a}\left(1+\frac\a2\right) &=e^{-\b_2}\left(\frac34+\frac{\b_2}2\right)~,  \\
	\b_1 &\leq \b_2~.
\end{align*}
The first two equations each implicitly defines a curve, and we prove these curves have a single point of intersection $(\a_0,\b_0)$. Therefore, the game $\fthotel(n,\r(\a))$ admits a Nash equilibrium if and only if $\a \geq \a_0$. Moreover, the canonical pair of this game is $(H(\a),M(\a))$, from which we construct its equilibrium $\X^{n,\r(\a)}$.
Finally, we show that $\a$ is monotonic in $\r$ and therefore the game $\fthotel(n,\r)$ has a Nash equilibrium if and only if $\r \geq \r(\a_0)$, and thus we obtain Theorem~\ref{thm:NE-lower-bound}.

\section{Efficiency of Equilibria}\label{sec:pos}
Hotelling's original paper illustrated the tension that exists between the stable configuration of servers and the socially optimal outcome, which takes into account also the interests of the clients. In the 2-player game, the only equilibrium is when both players locate themselves at the center. However, the socially optimal outcome for the clients is obtained in a configuration that places one player at $1/4$ and the other at $3/4$, as this configuration minimizes the average travel distance from a client to the closest server. In the 3-player game, Hotelling showed that no Nash equilibrium exists, so there is no stable configuration, let alone a socially optimal one.

\paragraph{Efficiency of Equilibria in Failure-Free Scenarios.}
Our model suggests that the presence of possible failure has also some unexpected \emph{positive} side-effects. In the 2-player game, depending on the failure rate $\r$, the distance between the players may increase while preserving stability, reducing $\Cfree(\X)$, the total transportation cost of clients (when no faults occur). Hence, with probability $e^{-\r}$, the system admits an instance where no faults occur, and the cost is reduced as a side effect of the players' reaction to faults.



In the 3-player game, the inclusion of random faults creates a Nash equilibrium, which does not exist otherwise. Here, the price of disconnecting some of the clients from servers is less problematic, since without faults the game has no stable configuration at all.


\begin{figure}[ht]
    \begin{center}
        \includegraphics[width=0.6\textwidth]{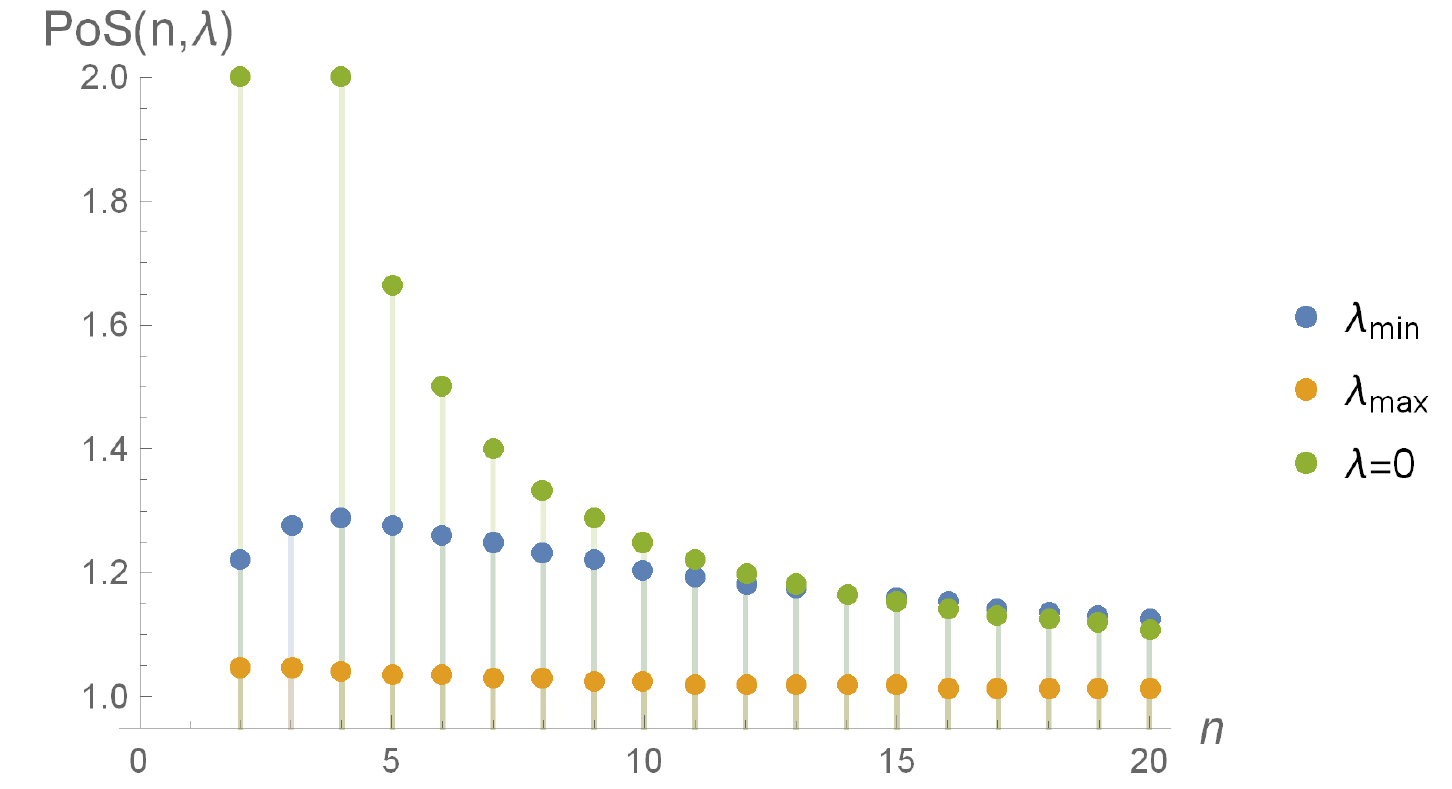}
    \end{center}
    \caption{The price of stability of the games $\fthotel(n,\r_{\min})$, $\fthotel(n,\r_{\max})$ and $\fthotel(n,0)$ under various $n$ values.}
    \label{fig:price-of-stability}
\end{figure}
For general $n$, we compare the price of stability obtained by the canonical profile for different values of $\r$. Namely, we consider the following three interesting values of $\r$ for every $n$.

\vspace{10pt}
\noindent
(1)
$\r = \r_{\min}$, as defined in Theorem~\ref{thm:NE-lower-bound}, which is the minimal value of $\r$ such that $\fthotel(n,\r)$ has a Nash equilibrium. Additionally, for $\r_{\min}$, the distance $M$ between neighboring players is the minimal such distance that allows an equilibrium (by Corollary~\ref{cor:NE-profiles}).

\vspace{10pt}
\noindent
(2)
$\r = \r_{\max}$, as defined at the end of Section~\ref{sec:canonical-profile}, which maximizes the possible distance $M$ between players in the canonical profile $\Xnr$ (by Corollary~\ref{cor:NE-profiles}).

\vspace{10pt}
\noindent
(3)
$\r = 0$, which, with a slight abuse of notation, refers to the Hotelling game with no faults.

\vspace{10pt}
The price of anarchy and price of stability for the fault-free Hotelling game $\fthotel(n,0)$ was studied by Fournier and Scarsini~\cite{fournier2016hotelling}. They showed that
$$
\hbox{$
\PoA(n,0) = \left\{\begin{array}{ll}
    						2 				& \quad \mbox{if $n$ is even,}\\
                            2\frac{n}{n+1}	& \quad \mbox{if $n>3$ is odd,}
    					\end{array}\right.
$}
~~~~~~~~\mbox{and}~~~~~~~~
\hbox{$
\PoS(n,0) = \left\{\begin{array}{ll}
    						n				& \quad \mbox{for $n=2$,} \\
                            \frac{1}{n-2}	& \quad \mbox{for $n\geq4$.}
    					\end{array}\right.
$}
$$
It is known that
$ C\left(\X^{\textsf{OPT}}\right)=1/(2n)$,
where $\X^{\textsf{OPT}}$ is the social optimum configuration.
For $\r \geq \r_{\min}$, by Theorems~\ref{thm:NE-form} and \ref{thm:NE-lower-bound}, the unique Nash equilibrium is the canonical profile, so
$$ \PoS(n,\r) = \PoA(n,\r) = \frac{\Cfree(\Xnr)}{\Cfree(\X^{\textsf{OPT}})}~. $$
Within a hinterland of length $H$, the average travel distance for a client is $H/2$, and within an internal region, the average travel distance is $M/4$. Hence,
	$$ \Cfree(\Xnr) = \frac{\frac{H}{2} \cdot 2H+\frac{M}{4} \cdot (n-1) M}{2H+(n-1)M}
     			= H^2 + \frac{n-1}{4}M^2~, $$
and thus
	$$ \PoS(n,\r) = \PoA(n,\r) = 2n \left( H^2 + \frac{n-1}{4}M^2 \right)~. $$
By Corollary~\ref{cor:NE-profiles}, since $\PoS(n,\r)$ depends only on $H$, $M$ and $n$, all possible values of $\PoS(n,\r)$ are represented in the range $ \r_{\min} \leq \r \leq \r_{\max} $. Moreover,
	$$ \PoS(n,\r_{\max}) \leq \PoS(n,\r) \leq \PoS(n,\r_{\min})~,  $$
since the average travel distance $\Cfree(\Xnr)$ decreases the larger $M$ is. Therefore, we compare\footnote{We do not include $\PoA(n,0)$ in this comparison because it is much larger than $\PoS(n,\r)$.} $\PoS(n,\r_{\max})$ and $\PoS(n,\r_{\min})$ to the non-faulty $\PoS(n,0)$ in Figure~\ref{fig:price-of-stability} above.

Figure~\ref{fig:price-of-stability} shows that $\PoS(n,\r_{\max})$ is close to 1, i.e., the client cost is close to optimal, and that $\PoS(n,\r_{\min})$ is clearly smaller than $\PoS(n,0)$ for $n \leq 13$. Note, however, that for $\r=\r_{\max}$, the probability of no fault occurring on the line is very small (less than $e^{-3n}$).

\paragraph{Efficiency of Equilibria Under Failures.}
So far we have considered how the clients are affected by the players
\emph{reactions} to possible failures. Next, we measure the effect of the failures \emph{themselves} on the costs incurred on the clients. To this end we quantify the cost to a disconnected client as $\discost$, and consider how this cost affects the clients in Nash equilibria of both the failure-free Hotelling game and the fault-prone Hotelling game proposed in this paper.

\begin{wrapfigure}{r}{0.35\textwidth}
\vspace{-10pt}
\centering
\includegraphics[width=0.35\textwidth]{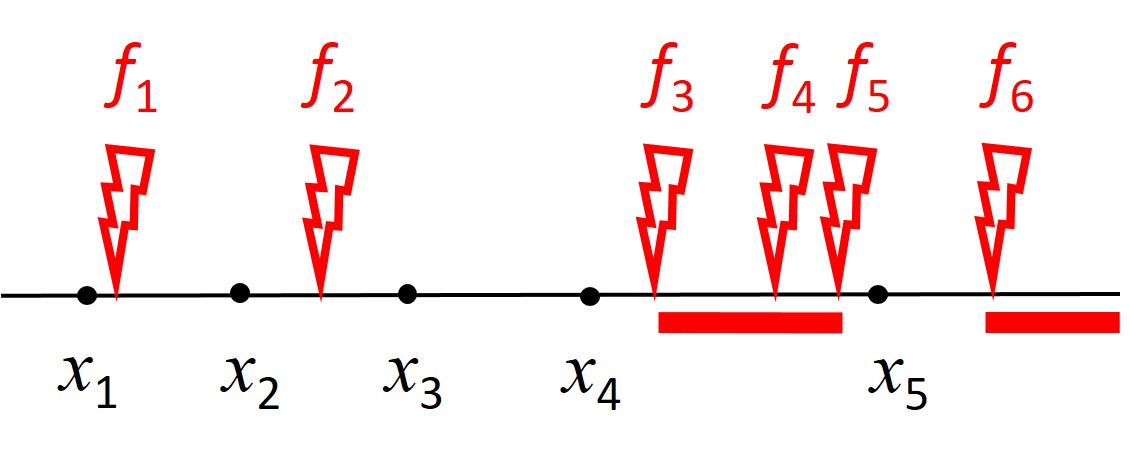}
\end{wrapfigure}
Note that the set of disconnected points consists of a finite number of
continuous segments (see figure).
Let $\cL^{dc}_\F(\X)$ denote the total length of these segments (hereafter referred to as the \emph{disconnected fraction}), and let
$\cL^{c}_\F(\X) = 1-\cL^{dc}_\F(\X)$ denote the total length of the complementary
segments, namely, the union of all the markets, containing the connected clients.

Notice that the transportation cost of each client $y\in Y^c_\F(\X)$ satisfies
$d(\X_\F(y),y) \le 1$, so their total transportation cost is
$C_\F^{c}(\X) \le \cL^c_\F(\X) \le 1$.
On the other hand, the total disconnection cost of the disconnected clients is
$C_\F^{dc}(\X) = \psi\cdot\cL^{dc}_\F(\X)$.
It follows that when $\discost\gg 1$,
the difference between different profiles in terms of their total access cost
under a fault set $\F$ is dominated by their disconnection cost, which in turn
is proportional to their $\cL^{dc}_\F(\X)$ value, namely,
the proportion of disconnected clients.

For given $\X$ and failure rate $\lambda$, let
$\cL^{dc}_\lambda(\X)$ denote the expected value of $\cL^{dc}_\F(\X)$
when $\F$ is distributed according to the Poisson process
with rate $\lambda$ discussed above.

We are interested in comparing the costs $\cL^{dc}_\lambda(\X)$
of three different profiles $\X$.
The first, denoted $\OPT$ in Figure~\ref{fig:welfare}, is the profile that
minimizes $\cL^{dc}_\lambda(\X)$. (This profile might not attain equilibrium.)
The second, denoted $\NE$ in Figure~\ref{fig:welfare},
is the canonical profile $\Xnr$ attaining Nash equilibrium in
the game $\fthotel(n,\r)$ (i.e., assuming failure rate $\lambda$).
The third, denoted $\FF$ in Figure~\ref{fig:welfare},
is the profile that minimizes $\cL^{dc}_\lambda(\X)$ from among the profiles $\X$ attaining Nash equilibrium in the game $\fthotel(n,0)$.

Note that $\OPT$ and $\NE$ depend on both $n$ and $\r$ while $\FF$ depends on $n$ alone. Hence, in each graph of Figure~\ref{fig:welfare}, $\FF$ relates to a fixed profile under different rates $\r$ of faults, while $\OPT$ and $\NE$ relate to different profiles for different values of $\r$.
\begin{figure}[b]
\centering
\begin{tabular}{ccc}
\includegraphics[width=0.3\textwidth]{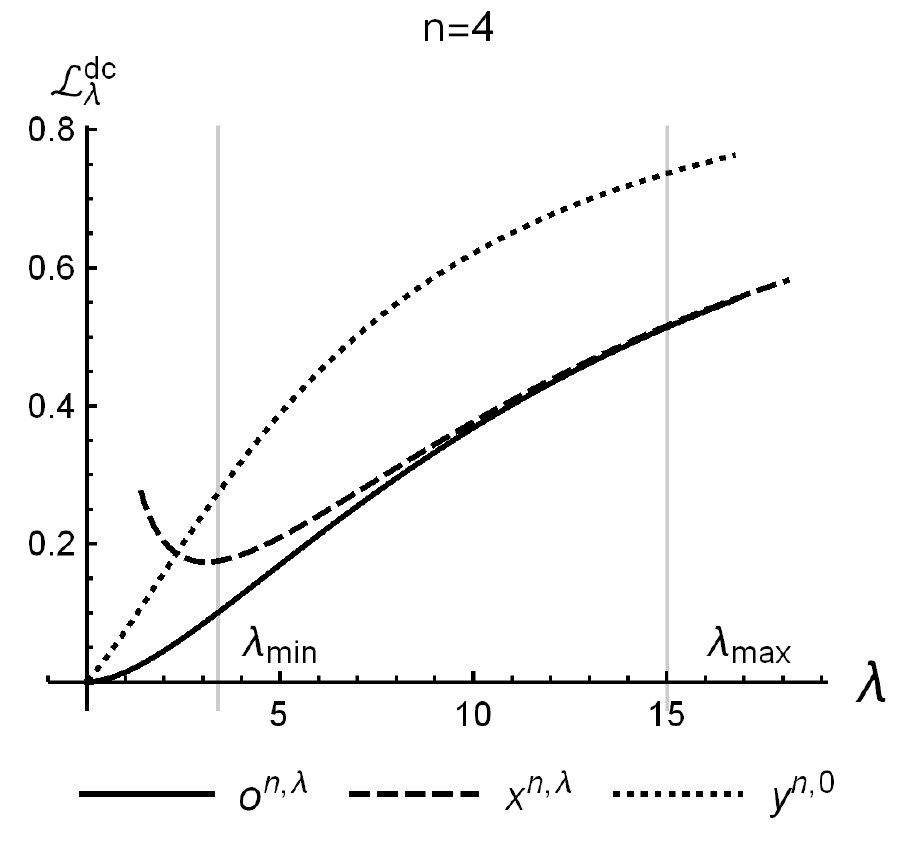}&
\includegraphics[width=0.3\textwidth]{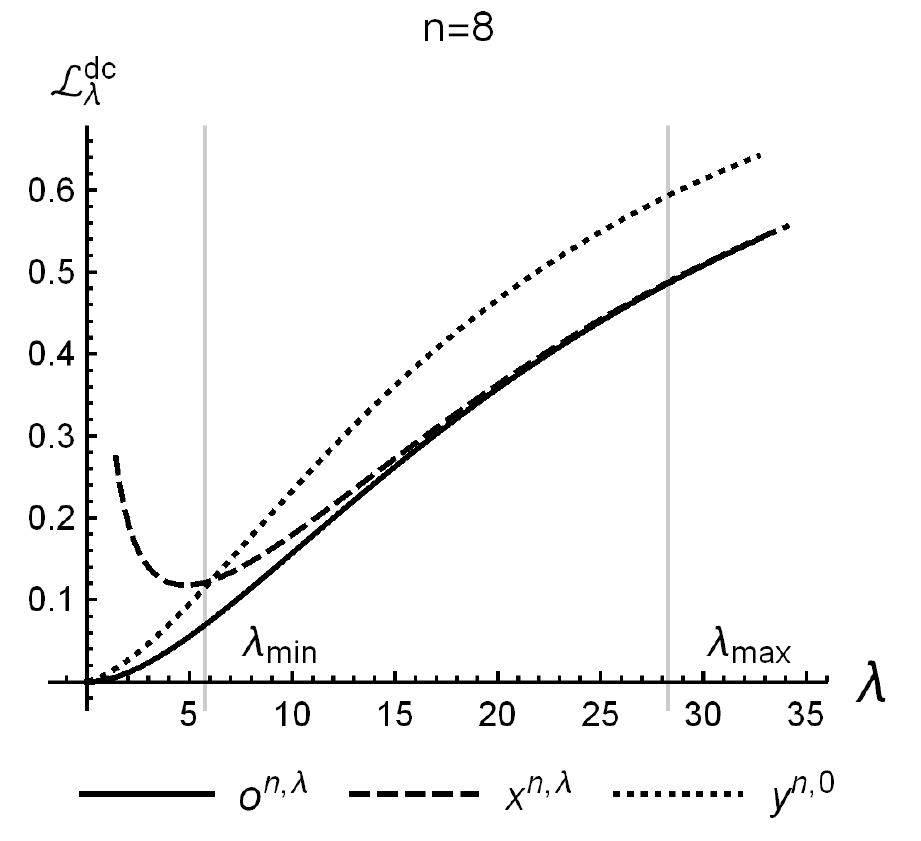}&
\includegraphics[width=0.3\textwidth]{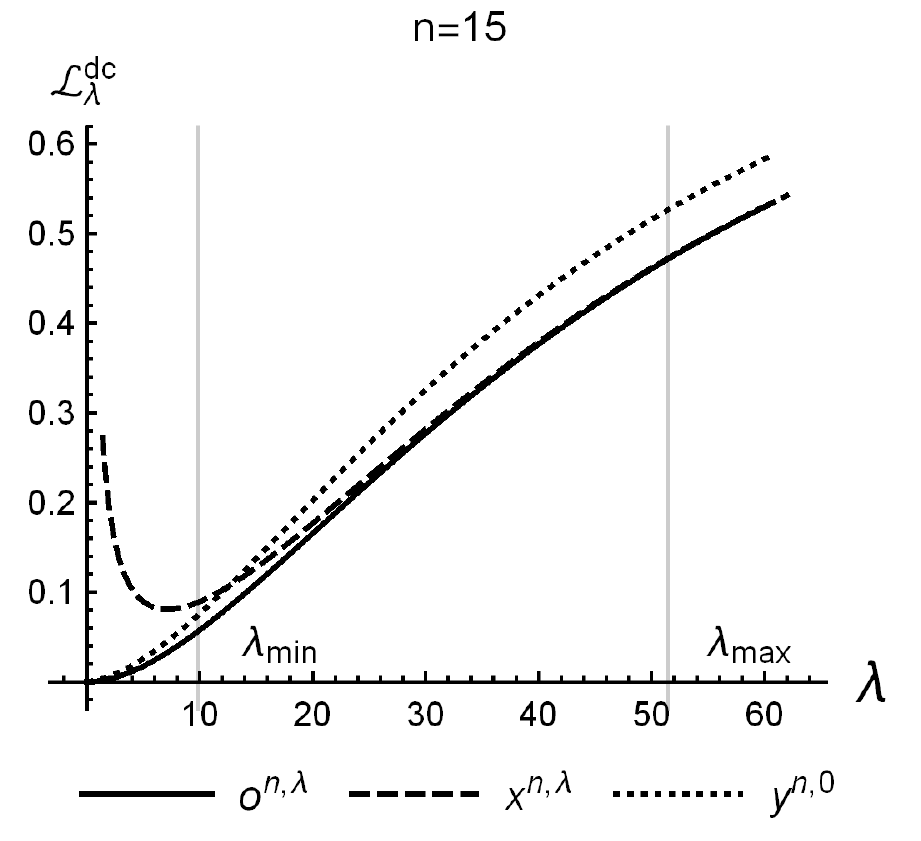}\\
		(a) & (b) & (c)
\end{tabular}
\caption{The disconnected fraction for the profiles $\OPT$, $\NE$ and $\FF$
under various $\lambda$ values.
(a) $n=4$.
(b) $n=8$.
(c) $n=15$.}
\label{fig:welfare}
\end{figure}

While we do not have closed-form formulae for the values of
$\cL^{dc}_\lambda(\OPT)$, $\cL^{dc}_\lambda(\NE)$ and \linebreak $\cL^{dc}_\lambda(\FF)$,
we can calculate these values for different values of $\lambda$ and $n$.
Figure~\ref{fig:welfare} presents graphically the comparisons for $n=4,8,15$.
As can be seen in the figure, our profile $\NE$ does better than $\FF$
(and thus better than all other Nash equilibria of $\fthotel(n,0)$ as well)
once the failure rate $\lambda$ increases, and its costs are close to those
of $\OPT$.

The reason for the gap between $\OPT$ and $\FF$ is that
in all equilibria of the failure-free game $\fthotel(n,0)$
the peripheral servers must be colocated with their neighbors.
This means that there are at most $n-2$ distinct server locations.
Thus, from the perspective of the clients there are two servers fewer in $\FF$ than in $\OPT$.
Indeed, simulation results show that the disconnected fraction under $\FF$ with $n$ players
closely approximates the disconnected fraction under $\bm{o}^{n-2,\r}$ with $n-2$ players
for all values of $\r$.

In conclusion, for sufficiently large $\r$ total client costs would be reduced under $\fthotel(n,\r)$ compared to the failure-free game whether faults occur or not. That is, clients would benefit from players reaction to possible faults even if no such possibility exists in actuality.

\section{Discussion}

Our results show that when players account for the possibility of faults in their behavior the properties of the game may be impacted in several interesting ways. In some settings a Nash equilibrium emerges where none existed before. In others, the change in player behaviors gives rise to a unique Nash equilibrium where infinitely many existed before. Moreover, in many settings the social cost is reduced when compared to the Nash equilibria of the original game (whether failures occur or not). Overall, our results imply player consideration of faults may introduce stability to the game, and in some cases even reduce social costs. Specifically in the context of Hotelling games, our results give a possible simple explanation for cases where the principle of minimal differentiation does not hold. That is, players prefer to clearly differentiate from their competitors rather than clustering together.

There are many possible directions for future work. First, one may examine whether our results generalize to different distributions of faults, such as other spatial point processes. Another interesting variation would be to have faults delay the passage of clients rather than disconnecting the line entirely. Second, a natural yet challenging extension of our work is when the game is played on networks. A vast literature exists on the fault tolerance of networks, and the motivation for such an extension is clear. Finally, our game may be extended to two-dimensional space. Inspired by models of wireless networks, instead of faults one may surround each player with a circle of a randomly distributed radius, within which it obtains the clients closer to it than other players. This models the range of influence of each player, which may be limited due to ineffective marketing for example. It would be interesting to discover if such games admit an equilibrium, since no equilibrium exists for failure-free Hotelling games in two-dimensional space with more than two players.

\section*{Acknowledgements}
The authors would like to thank Yinon Nahum for many fertile discussions and helpful insights.


\bibliographystyle{plain}




\clearpage
\centerline{\LARGE \bf Appendix}

\appendix

\section{The Poisson Process}\label{sec:poisson-process}
In this section we present a few well known results that are necessary for our proofs. Namely, we discuss the Poisson process, which we use as our fault model. The following definition and theorems are taken from Mitzenmacher and Upfal~\cite{mitzenmacher2005probability}.

The Poisson process is a stochastic counting process which considers a sequence of random events. Let $Q(t)$ denote the number of events in the interval $[0,t]$. The locations of these events are called \emph{arrival times}. Formally, a \emph{Poisson process} with parameter $\lambda$ is a stochastic counting process $\{Q(t),t \geq 0\}$ such that the following statements hold.
\begin{enumerate}
    \item $Q(0)=0$.
    \item The process has independent and stationary increments. That is, for any $t,s>0$, the distribution of $Q(t+s)-Q(s)$ is identical to the distribution of $Q(t)$, and for any two disjoint intervals $[t_1,t_2]$ and $[t_3,t_4]$, the distribution of $Q(t_2)-Q(t_1)$ is independent of the distribution of $Q(t_4)-Q(t_3)$.
    \item $\lim_{t \to 0}{\Prb(Q(t)=1)/t}=\lambda$. That is, the probability of a single event in a short interval tends to $\lambda t$.
    \item $\lim_{t \to 0}{\Prb(Q(t)\geq 2)}=0$. That is, the probability of more than one event in a short interval tends to zero.
\end{enumerate}

\begin{theorem}\label{thm:poisson-dist}
    Let $\{Q(t)\mid t\geq 0 \}$ be a Poisson process with parameter $\lambda$. For any $t,s \geq 0$ and any integer $k \geq 0$,
    \[
        \Prb(Q(t+s)-Q(s)=k)= e^{-\lambda t}\frac{(\lambda t)^n}{n!}~.
    \]
\end{theorem}

\begin{theorem}\label{thm:poisson-uniform}
    Given that $Q(t)=k$, the $k$ arrival times have the same distribution as that of $n$ independent random variables with uniform distribution over $[0,t]$.
\end{theorem}

\section{Analysis}\label{sec:analysis}

\subsection{Calculating Payoffs}\label{sec:region-profit}

In this section we focus on calculating the payoffs of each player for any given profile.

We start by analyzing two simple cases which can be easily extended to the general setting. Suppose the game is played on the interval $[0,d]$ for $d > 0$ and player 1 is located at $x_1=0$. $\H(d)$ is a random function indicating the profit of player 1 when there are no other players in the game, i.e., $[0,d]$ is a hinterland. $\M(d)$ is a random function indicating the profit of player 1 when the game has two players and player 2 is located at $x_2=d$, i.e., $[0,d]$ is an internal region. Additionally, $Q(d)$ is a random function indicating the number of faults in the interval $[0,d]$. Hence, $Q(d)$ is a Poisson process with parameter $\r$ as defined in Section~\ref{sec:poisson-process}. It is easy to calculate
the expectations of $\H(d)$ and $\M(d)$ restricted to the case where
$k$ faults occur in the line segment. Clearly, when no faults occur we have
\begin{equation}\label{eq:zero-faults}
\Ebb[\H(d)\mid Q(d)=0]=d \quad \mbox{ and } \quad  \Ebb[\M(d)\mid Q(d)=0]=\nicefrac{d}{2}~.
\end{equation}
For $Q(d)>0$ we have the following.
\begin{claim}\label{clm:k-uniform-faults}
For $k>0$ and $d>0$,
\[
\Ebb[\H(d)\mid Q(d)=k]=\Ebb[\M(d)\mid Q(d)=k]=\frac{d}{k+1}~.
\]
\end{claim}
\begin{proof}
By Theorem~\ref{thm:poisson-uniform}, given that $Q(d)=k$, the $k$ faults have the same distribution as $k$ independent random variables with uniform distribution over $[0,d]$. Furthermore, since $k>0$, both $\H(d)$ and $\M(d)$ are simply the distance to the closest fault. So our question boils down to the following: suppose $k>0$ faults $Y_1,\ldots,Y_k$ are uniformly distributed on the interval $[0,d]$, what is the expectation of $\min(Y_i)$? Formally we have
$$
\Ebb[\H(d)\mid Q(d)=k]=\Ebb[\min(Y_i)]=\int_0^\infty{\Prb[\min(Y_i)>t] \, dt}=\int_0^d{\Prb[\min(Y_i)>t] \, dt}~.
$$
Since the $Y_i$'s are independent and uniformly distributed in the interval $[0,d]$ we get
$$
\Prb[\min(Y_i)>t]= \prod_i\Prb[Y_i > t]= \Big(\frac{d-t}{d}\Big)^k~,
$$
and therefore
$$
    \Ebb[\H(d)\mid Q(d)=k] =\int_0^d{\Big(\frac{d-t}{d}\Big)^k \, dt}
        =\left[-\frac d{k+1}\Big(\frac{d-t}{d}\Big)^{k+1}\right]_0^d =\frac d{k+1}~.
$$
\end{proof}
By Theorem~\ref{thm:poisson-dist}, for all $d\geq 0$ and $k\geq0$,
\begin{equation}\label{eq:k-faults-prob}
    \Prb[Q(d)=k]=e^{-\r d}\frac{(\r d)^k}{k!}~.
\end{equation}

In sake of brevity, for any random variable $X$, we define $\mean{X}=\Ebb[X]$.
We may now find $\EH{d}=\Ebb[\H(d)]$ and $\EM{d}=\Ebb[\M(d)]$ using the law of total probability.
\begin{claim}\label{clm:region-profit}
   $\displaystyle \quad \EH{d}=\frac{1}{\r} \left[1-e^{-\r d}\right] \quad \mbox{and } \quad
      \EM{d}=\frac{1}{\r} \left[1-e^{-\r d}\left(1+\frac{\r d}{2}\right)\right]~.  $
\end{claim}
\begin{proof}
    By Claim~\ref{clm:k-uniform-faults} and Equations~\eqref{eq:zero-faults} and \eqref{eq:k-faults-prob}, we get
    \begin{align*}
        \Ebb[\H(d)] & =\sum_{k=0}^{\infty}\left(\Ebb[\H(d)\mid Q(d)=k]\cdot\Prb[Q(d)=k]\right)
         \; =\sum_{k=0}^{\infty}\left(\frac{d}{k+1}\cdot e^{-\r d}\frac{(\r d)^k}{k!}\right)\\
         & =e^{-\r d} \cdot \sum_{k=0}^{\infty}\frac{\r^k d^{k+1}}{(k+1)!}
         \; =\frac{e^{-\r d}}{\r} \cdot \sum_{k=1}^{\infty}\frac{(\r d)^k}{k!}
         \; =\frac{e^{-\r d}}{\r} \left[\sum_{k=0}^{\infty}\frac{(\r d)^k}{k!}-1\right]\\
         & =\frac{e^{-\r d}}{\r} \left[e^{\r d}-1\right]
         \; =\frac{1}{\r} \left[1-e^{-\r d}\right]~. \numberthis \label{eq:exp-hinterland}
    \end{align*}
    Similarly, by Claim~\ref{clm:k-uniform-faults},
    \begin{align*}
        \Ebb[\M(d)] & =\sum_{k=0}^{\infty}\left(\Ebb[\M(d)\mid Q(d)=k]\cdot\Prb[Q(d)=k]\right)\\
            & =\Ebb[\M(d)\mid Q(d)=0]\cdot\Prb[Q(d)=0]
                        +\sum_{k=1}^{\infty}\left(\Ebb[\H(d)\mid Q(d)=k]\cdot\Prb[Q(d)=k]\right)~.
    \end{align*}
    Rearranging and applying Equations~\eqref{eq:zero-faults},\eqref{eq:k-faults-prob} and \eqref{eq:exp-hinterland}, we get
    \begin{align*}
         \Ebb[\M(d)] & =(\Ebb[\M(d)\mid Q(d)=0]-\Ebb[\H(d)\mid Q(d)=0])\cdot\Prb[Q(d)=0]+\Ebb[\H(d)]\\
         & =\Ebb[\H(d)]-\frac{d}{2}\,\Prb[Q(d)=0]
           =\frac{1}{\r} \left[1-e^{-\r d}\right]-\frac{d}{2}\,e^{-\r d}\\
         & =\frac{1}{\r} \left[1-e^{-\r d}\left(1+\frac{\r d}{2}\right)\right]~.
    \end{align*}
\end{proof}

Next, consider a profile $\X$ of the game. We may now calculate the payoff of each player using the analysis above. That is, we have the following.
\begin{claim}\label{clm:region-equiv}
    For every $i\in N$,
    \[
        \Ebb[R_i-x_i]=
                \left\{\begin{array}{ll}
                            \EH{1-x_i}~, & \mbox{if $i$ is right peripheral;} \\
                            \EM{x_i^r-x_i}~, & \mbox{otherwise,}
                        \end{array}
        \right.
    \]
        and
    \[
        \Ebb[x_i-L_i]=
                \left\{\begin{array}{ll}
                            \EH{x_i}~, & \mbox{if $i$ is left peripheral;} \\
                            \EM{x_i-x_i^\ell}~, & \mbox{otherwise.}
                        \end{array}
        \right.
    \]
\end{claim}
\begin{proof}
    Consider $\Ebb[R_i-x_i]$ first. Let $d=1-x_i$ if $i$ is right peripheral and $d=x_i^r-x_i$ otherwise.
    Since $Q(t)$ is a Poisson process, by definition, the number of faults within the interval $[x_i,x_i+d]$ is $Q(x_i+d)-Q(x_i)$ and has the same distribution as $Q(d)$. Hence, for all $k\geq 0$
    \begin{equation}\label{eq:shifted-prob}
        \Prb[Q(x_i+d)-Q(x_i)=k]=\Prb[Q(d)=k]~.
    \end{equation}
    Moreover, as in Equation~\eqref{eq:zero-faults}
    \begin{equation}\label{eq:zero-faults2}
        \Ebb[R_i-x_i \mid Q(x_i+d)-Q(x_i)=0]=\left\{\begin{array}{ll}
                                                    d~, & \mbox{if $i$ is right peripheral;} \\
                                                    \frac{d}{2}~, & \mbox{otherwise.}
                                                \end{array}
                                             \right.
    \end{equation}
    Finally, notice that when $k>0$ faults occur in $[x_i,x_i+d]$, by Theorem~\ref{thm:poisson-uniform}, they are uniformly distributed in the interval. Therefore, as in Claim~\ref{clm:k-uniform-faults} we obtain
    \begin{equation}\label{eq:k-uniform-faults2}
        \Ebb[R_i-x_i\mid Q(x_i+d)-Q(x_i)=k]=\frac{d}{k+1}~.
    \end{equation}
    Using the law of total probability and applying Equations~\eqref{eq:shifted-prob},\eqref{eq:zero-faults2} and \eqref{eq:k-uniform-faults2} we obtain the same terms as in the proof of Claim~\ref{clm:region-profit} and thus the result follows.

    The proof for $\Ebb[x_i-L_i]$ is similar. Note only that the analogue of Equation~\eqref{eq:k-uniform-faults2} follows from the symmetry of the uniform distribution.
\end{proof}

\subsection{Optimizing the Payoff Within an Interval}\label{sec:local-optimum}

As we show in this section, given a profile $\X$, for each server $i$ there exists a single point of maximum payoff within the interval between its neighbors.  Furthermore, when each server $i$ is located at that optimal point between its neighbors, the profile is the canonical profile. More specifically, by Lemma~\ref{lem:opt-internal}, for any internal server $i$, the optimal location is at the center of the interval $[x_i^\ell,x_i^r]$. Therefore, all internal regions are of equal length. For peripheral servers, by Lemma~\ref{lem:opt-hinterland}, the optimal location can be derived from Equation~\eqref{eq:opt-hinterland}.

Denote by $\theta(t,s)$ the expected market of a peripheral server with a hinterland of length $t$ and a neighbor at distance $s$ from the border of the line. By Claim~\ref{clm:region-profit},
$$\theta(t,s)= \EH{t}+\EM{s-t} = \frac{1}{\r} \left[1-e^{-\r t}\right]
            + \frac{1}{\r} \left[1-e^{-\r (s-t)}\left(1+\frac{\r (s-t)}{2}\right)\right]$$
for $ s\in[0,1] $ and $t \in [0,s]$. Denote by $\mu(t,s)$ the expected market of an internal server $i$ with $x_i-x_i^\ell=t$ and $x_i^r-x_i^\ell=s$. By Claim~\ref{clm:region-profit},
$$\mu(t,s)= \EM{t}+\EM{s-t} = \frac{1}{\r} \left[1-e^{-\r t}\left(1+\frac{\r t}{2}\right)\right]
            + \frac{1}{\r} \left[1-e^{-\r (s-t)}\left(1+\frac{\r (s-t)}{2}\right)\right]$$
for $ s\in[0,1] $ and $t \in [0,s]$.

\begin{remark}
    To keep $\theta$ and $\mu$ continuous in the interval $[0,s]$, we disregard the fact that for $t=s$ and $t=0$ the server is colocated with one of its neighbors, and assume all its payoff comes from the same interval of length $s$. As we show in Claim~\ref{clm:no-colocated}, this assumption does not affect the analysis of the Nash equilibria of the game.
\end{remark}

\begin{lemma}\label{lem:opt-hinterland}
    For every fixed $s\in[0,1]$ there exists a single optimal point $t^*\in[0,s]$ such that $\theta(t^*,s)=\max_t\theta(t,s)$. Moreover, if $\r s\leq \ln2$ then $\theta(t,s)$ is monotone increasing as a function of $t$ and thus $t^*=s$, and if $\r s > \ln2$ then $t^*$ satisfies the equation
    \begin{equation}\label{eq:opt-hinterland}
			e^{\lambda(s-2t^*)} = \frac{1+\lambda(s-t^*)}2~.
	\end{equation}
    Additionally, if $\r s > \ln2$ then $\r t^* > \ln 2$ as well.
\end{lemma}
\begin{proof}
    Note that
    $$ \frac{\partial^2\,\theta(t,s)}{\partial t^2}=-\r e^{-\r t}-\frac{\r^2(s-t)}2e^{-\r(s-t)}<0~, $$
    for all $t\in[0,s]$. Hence, $\theta(t,s)$ is strictly concave as a function of $t$ and thus has a single maximum point. The first derivative is
    $$ \frac{\partial\,\theta(t,s)}{\partial t} = e^{-\r t}-\frac12 e^{-\lambda (s-t)}(1+\lambda (s-t))~.  $$
    Setting $t=0$ yields
    $$ \frac{\partial\,\theta(0,s)}{\partial t} = 1-\frac12 e^{-\lambda s}(1+\lambda s)~,  $$
    which is positive for all $s$ since $e^x>1+x$ for all real $x$.

    \noindent Setting $t=s$ yields
    $$ \frac{\partial\,\theta(s,s)}{\partial t} = e^{-\lambda s}-\frac12~, $$
    which is non-negative whenever $\r s \leq \ln2$. Since $\theta(t,s)$ is strictly concave, it follows that whenever $\r s \leq \ln2$ the function $\theta(t,s)$ is monotone increasing in $t$ for $t\in[0,s]$. Therefore, it has a single maximum point at $t^*=s$.

    It is left to consider the case where $\r s > \ln2$. In this case, the maximum is obtained where the derivative equals zero. Hence,
    $$ \frac{\partial \,\theta(t^*,s)}{\partial t} = e^{-\r t^*}-\frac12 e^{-\lambda (s-t^*)}(1+\lambda (s-t^*))=0~,  $$
    which translates to Equation~\eqref{eq:opt-hinterland}, as required.

    Finally, we show that in this case $\r t^* > \ln 2$. Assume first, towards contradiction, that $\r t^* < \ln 2$. Then
    \[
        e^{\r(s-2t^*)}=\frac{e^{\r(s-t^*)}}{e^{\r t^*}}>\frac{e^{\r(s-t^*)}}{e^{\ln 2}}\geq \frac{1+\r(s-t^*)}2~,
    \]
    contradicting Equation~\eqref{eq:opt-hinterland}. Now suppose $\r t^*=\ln2$. Then
    \[
        \frac{e^{\r(s-t^*)}}{e^{\ln 2}} = e^{r(s-2t^*)} = \frac{1+\r(s-t^*)}2~,
    \]
    or
    \[
        e^{\r(s-t^*)} = 1+\r(s-t^*)~,
    \]
which holds if and only if $\r(s-t^*)=0$ and thus $\r s = \r t^* =\ln2$, which contradicts the assumption that $\r s > \ln2$. This proves that $\r t^* >\ln2$ whenever $\r s > \ln2$, concluding the proof of the lemma.
\end{proof}

\begin{lemma}\label{lem:region-difference}
    Let $s\in [\ln2/\r,1]$ , and let $t^*$ be as in Lemma~\ref{lem:opt-hinterland}. Then,
        $$ \EH{t^*}=\EM{s-t^*}+\frac{1}{2\r}e^{-\r(s-t^*)} $$
\end{lemma}
\begin{proof}
    By Claim~\ref{clm:region-profit},
    $$ \EH{t^*}= \frac1\r\left[1-e^{-\r t^*}\right]~. $$
    Plugging in Equation~\eqref{eq:opt-hinterland}, we get
    \begin{align*}
        \EH{t^*}   &= \frac1\r\left[1-e^{-\r (s-t^*)}\left(\frac{1+\r(s-t^*)}2\right)\right] \\
                        &= \frac1\r\left[1-e^{-\r (s-t^*)}\left(1+\frac{\r(s-t^*)}2\right)\right]
                                                +\frac{1}{2\r}e^{-\r(s-t^*)}
    \end{align*}
    Plugging in Claim~\ref{clm:region-profit} we obtain the result.
\end{proof}

\begin{claim}\label{clm:center-opt}
    Let $f:[0,1]\to \mathbb{R}$ be a strictly concave function, and define $g(t)= f(t)+f(s-t)$ for some $s\in[0,1]$. Then,
        $$g\left(\frac{s}{2}\right)=\max_{t\in[0,s]}{g(t)}~.$$
\end{claim}
\begin{proof}
    Note that
        $$ g'(t) = f'(t)-f'(s-t), $$
    and thus $g'(s/2) = 0$.
    Additionally, we have that $f''(t)<0$ for $t\in [0,1]$. Therefore,
        $$ g''(t) = f''(t)+f''(s-t) < 0 $$
    for $t\in [0,s]$. Hence, $s/2$ maximizes $g$ in the interval $[0,s]$.
\end{proof}

\begin{lemma}\label{lem:opt-internal}
    For every fixed $s\in[0,1]$ there exists a single optimal point $t^*\in[0,s]$ such that $\mu(t^*,s)=\max_t\mu(t,s)$. Moreover, $t^*=s/2$ for all $s\in[0,1]$.
\end{lemma}
\begin{proof}
    $\EM{t}$ is strictly concave as a function of $t$, since
    $$ \frac{d^2\,\EM{t}}{d t^2}=-\frac{\r^2 t}2 e^{-\r t}<0~, $$
    for every $t\in [0,s]$. Therefore, the lemma follows from Claim~\ref{clm:center-opt} and the definition of $\mu(t,s)$.
\end{proof}

\begin{observation}\label{obs:utility}
	For every internal server $i\in N$:
	$ \displaystyle
		u_i(\X)= \frac{\mu(x_i-x_i^\ell,x_i^r-x_i^\ell)}{\gamma_i}~.
	$

	For every left peripheral server $i \in N$:
	$ \displaystyle
		u_i(\X)= \frac{\theta(x_i,x_i^r)}{\gamma_i}~.
	$

	For every right peripheral server $i \in N$:
	$ \displaystyle
		u_i(\X)= \frac{\theta(1-x_i,1-x_i^\ell)}{\gamma_i}~.
	$
\end{observation}

\subsection{Nash Equilibria of the Game}\label{sec:NE}
We next consider what kinds of profiles can be Nash equilibria in our game. If $n=1$, i.e., only one server is located on the segment, then by Claim~\ref{clm:center-opt}, its unique optimal location is at the center, $x_1=1/2$. If $n=2$, then both servers are peripheral and therefore Lemma~\ref{lem:opt-hinterland} is sufficient to find all the Nash equilibria of the game. Hence, we next prove Theorem~\ref{thm:NE-2players}.

\par\noindent {\em Proof of Theorem~\ref{thm:NE-2players}.~}
    Let $(x_1,x_2)$ be a Nash equilibrium of the game and assume w.l.o.g. that $x_1 \leq x_2$. First note that $x_1\leq 1/2$, otherwise server 2 can improve its payoff by relocating to $1-x_2$ (due to the monotonicity of the payoff from an internal region). Symmetrically, we also have $x_2\geq 1/2$.

    Consider the case where $\r \leq 2\ln2$, and suppose that $x_2 > 1/2$. Hence, $ \r x_2 > \ln2 $ and by Lemma~\ref{lem:opt-hinterland}, it follows that $ \r x_1 > \ln2 $ when $x_1$ is optimally located in the segment $[0,x_2]$. Hence, $x_1>1/2$, contradiction. So suppose $x_2 = 1/2$ and thus $ \r x_2 \leq \ln2 $. It follows, by Lemma~\ref{lem:opt-hinterland}, that the optimal location for server 1 is at $x_1=1/2$.

    Now consider the case where $\r > 2\ln2$. It follows that $\r x_2 > \ln2$ and $ \r (1-x_1) > \ln2 $, and thus by Lemma~\ref{lem:opt-hinterland}, we get
	\[
		e^{\r(x_2-2x_1)} = \frac{1+\r (x_2-x_1)}2
	       \qquad  \hbox{ and }  \qquad
		e^{\r((1-x_1)-2(1-x_2))} = \frac{1+\r (x_2-x_1)}2~,
	\]
    which yields $e^{\r(x_2-2x_1)} = e^{\r((1-x_1)-2(1-x_2))}$, and therefore $x_1=1-x_2$. It follows that $(x_1,x_2)$ is the canonical profile $\X^{2,\r}$. Clearly, no improving move exists for either server so this is a Nash equilibrium.
\qed


Finally, consider the case where there are at least three servers in the game, i.e., $n\geq3$. First note that in every Nash equilibrium, every server is isolated, as we show in the following claim.

\begin{claim}\label{clm:no-colocated}
    If $n\geq3$, then in every Nash equilibrium, $\gamma_i=1$ for all $i\in N$.
\end{claim}
\begin{proof}
    Assume towards contradiction that $\X$ is a Nash equilibrium and there exist servers $i,i+1,\ldots,i+k$ such that $x_i=x_{i+1}=\cdots=x_{i+k}$ for $k>1$. Hence,
    \[
        u_i(\X)=\frac{\Ebb[R_i-L_i]}{k+1}~.
    \]
    But $\Ebb[R_i-L_i]=\Ebb[R_i-x_i]+\Ebb[x_i-L_i]$ and therefore $$\max(\Ebb[R_i-x_i],\Ebb[x_i-L_i])>\frac{\Ebb[R_i-L_i]}2>\frac{\Ebb[R_i-L_i]}{k+1}$$
    Assume w.l.o.g. that $\Ebb[R_i-x_i] \geq \Ebb[x_i-L_i]$. Clearly, relocating to $x_i+\varepsilon$ for sufficiently small $\varepsilon>0$ would improve the payoff of server $i$, contradicting the assumption that $\X$ is a Nash equilibrium. This proves that $\gamma_i\leq2$ for all $i\in N$.

    Moreover, even for $k=1$ such an improving move exists unless $\Ebb[R_i-x_i]=\Ebb[x_i-L_i]=\Ebb[R_i-L_i]/2$. So it is left to consider the case where $\gamma_i=2$ and $\Ebb[R_i-x_i]=\Ebb[x_i-L_i]$. At least one of the regions incident to $x_i$, w.l.o.g. $[x_i,x_i^r]$, is an internal region. Therefore, $\Ebb[R_i-x_i]=\EM{0}+\EM{x_i^r-x_i}=\mu(0,x_i^r-x_i)$. But by Lemma~\ref{lem:opt-internal}, $ \mu(0,x_i^r-x_i)<\mu((x_i^r-x_i)/2,x_i^r-x_i) $, i.e., relocating to the center of the interval $[x_i,x_i^r]$ improves the payoff of server $i$, contradiction. This proves that $\gamma_i=1$ for all $i\in N$.
\end{proof}

We are now ready to prove Theorem~\ref{thm:NE-form}.

\par\noindent {\em Proof of Theorem~\ref{thm:NE-form}.~}
	Let the profile $\X$ be a Nash equilibrium. By Claim~\ref{clm:no-colocated}, no pair of servers shares the same location. So we may assume w.l.o.g. that $0 \leq x_1 < x_2 < \ldots < x_n \leq 1$.
	Let $M=x_2-x_1$, $H_1=x_1$ and $H_2=1-x_n$. By Lemma~\ref{lem:opt-internal} we have that $M=x_2-x_1=x_3-x_2=x_4-x_3=\cdots=x_n-x_{n-1}$. Substituting $H_1$, $H_2$ and $M$ into Equation~\eqref{eq:opt-hinterland} we get
	\[
		e^{\r(M-H_1)} = \frac{1+\r M}2
	       \qquad  \hbox{ and }  \qquad
		e^{\r(M-H_2)} = \frac{1+\r M}2
	\]
	which yields $ e^{\r(M-H_1)} = e^{\r(M-H_2)}$ and therefore $H_1=H_2=H$. This proves that $x_i=H+(i-1)M$ for $i\in N$, and since $1-H=x_n=H+(i-1)M$, Equation~\eqref{eq:total-regions} follows as well.
	
	Finally, Equation~\eqref{eq:optimal-H} follows immediately by plugging $H$ and $M$ into Equation~\eqref{eq:opt-hinterland}, and Equation~\eqref{eq:minimal-lambda} follows from the last part of Lemma~\ref{lem:opt-hinterland}. This concludes the proof of the theorem.
\qed

\subsection{Properties of the Canonical Profile}\label{sec:canonical-profile}
Consider the canonical profile $\Xnr$. Then by Observation~\ref{obs:utility}, the payoff of the peripheral players is
\[
    u_1(\Xnr)=u_n(\X^{n,\r})=\theta(H,H+M)=\EH{H}+\EM{M}~,
\]
and the payoff of each internal player $i\in\{2,\ldots,n-1\}$ is
\[
    u_i(\Xnr)=\mu(M,2M)=\EM{M}+\EM{M}~.
\]
Additionally, by Lemma~\ref{lem:region-difference}, the payoff of a peripheral player always exceeds the payoff of an internal player by \emph{exactly} $e^{-\r M}/2\r$, i.e.,
\[
    u_1(\Xnr)-u_2(\Xnr)=\frac{e^{-\r M}}{2\r}~.
\]

Before we continue, we define a reparamatrization that would be useful to us in the following proofs. Define
\begin{equation}
\label{eq:def-alpha}
\a=\r M   ~~~~\mbox{and}~~~~  c=1-H/M~.
\end{equation}

\begin{claim}\label{clm:reparam}
    The values of $c$, $M$, $H$ and $\r$ can be expressed in terms of $\a$ and $n$ as follows.
    \begin{description}
      \item{(1)} $c=\frac{\ln\left(\frac{1+\a}{2}\right)}{\a}$; 
      \item{(2)} $M=\frac{1}{n+1-2c}$;                          
      \item{(3)} $H=\frac{1-c}{n+1-2c}$;                        
      \item{(4)} $\r=\a(n+1)-2\ln\left(\frac{1+\a}{2}\right)$.  
    \end{description}
\end{claim}
\begin{proof}
Substituting $c$ and $\a$ into Equation~\eqref{eq:optimal-H} we obtain $ e^{c\a} = (1+\a)/2$,
which yields the first part of the claim.
Plugging $H=(1-c)M$ into Equation~\eqref{eq:total-regions} yields the next two parts.
To obtain the last part we plug the terms we obtained for $M$ and $c$ into $\r=\a/M$.
\end{proof}

\begin{observation}\label{obs:monotone-lambda}
    The parameter $\r$ is monotone increasing as a function of $\a$ for all $\a>0$ and $n>2$. Therefore, $\a$ is a monotone increasing function of $\r$ as well.
\end{observation}

The observation follows from the fact the $\a$ derivative of $\r$ is strictly positive.

Keeping $n$ fixed, by Claim~\ref{clm:reparam}, for each $\a>0$ we obtain $\r$, $H$ and $M$, such that $H$ and $M$ are the canonical pair of $\fthotel(n,\r)$. Moreover, by Observation~\ref{obs:monotone-lambda}, considering $\r$ as a function of $\a$ over the domain $\a\in(0,\infty)$, $\r$ obtains all the values $\r\in(2\ln2,\infty)$. In conclusion, we have the following.

\begin{corollary}\label{cor:canonical-existence}
    The canonical profile $\Xnr$ exists for all $n\geq2$ and $\r>2\ln2$.
\end{corollary}

Additionally, while we do not have an explicit representation of $M$ and $H$ as functions of $\r$, we may examine the parametric curves $(\r(\a),M(\a))$ and $(\r(\a),H(\a))$ for $\alpha>0$ (see Figure~\ref{fig:param-M-H}).
\begin{figure}[ht]
    \begin{center}
        \includegraphics[height=6cm]{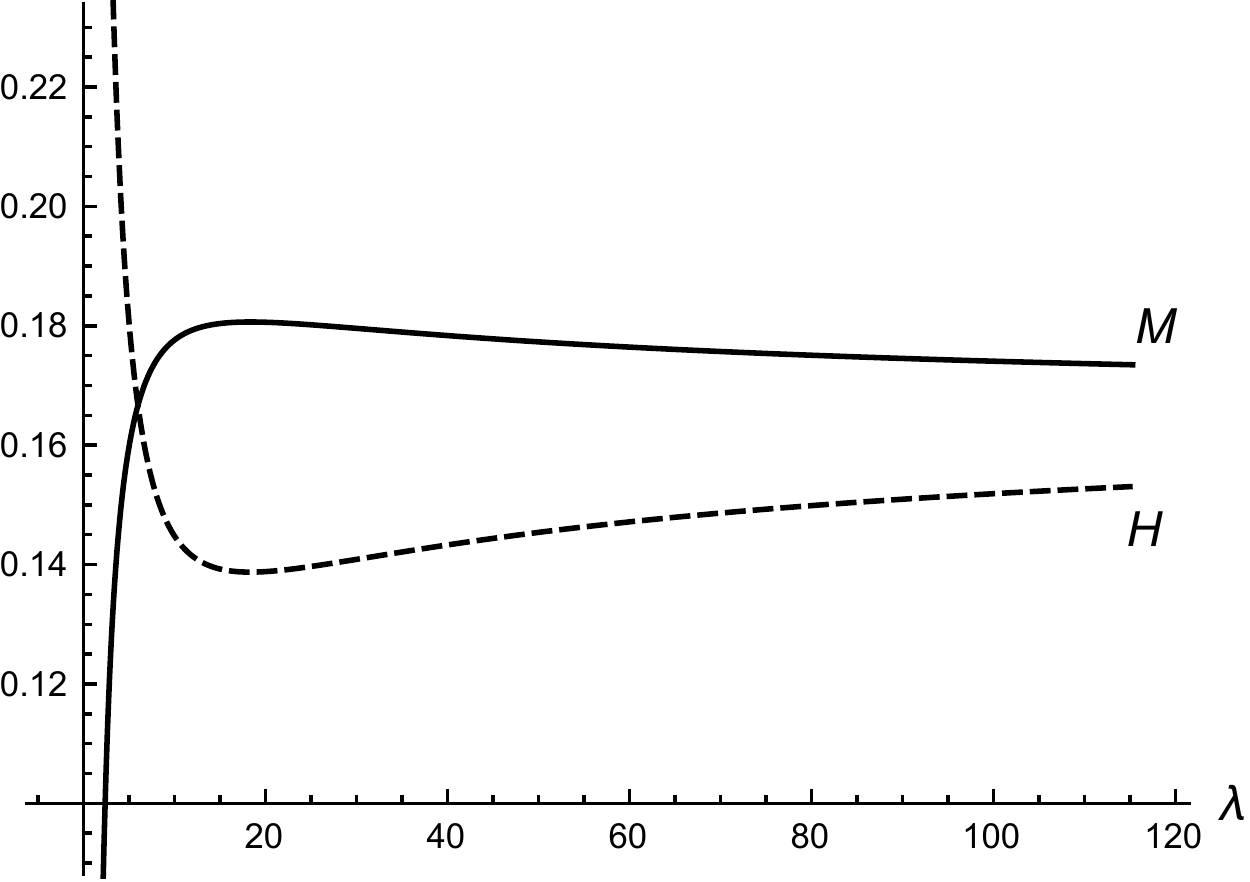}
    \end{center}
    \caption{parametric curves of $M$ and $H$ as a function of $\r$ ($n=5$).}
    \label{fig:param-M-H}
\end{figure}

The curves in Figure~\ref{fig:param-M-H} depend on the value of $n$. To examine the general behavior of the canonical profile, consider $c$ as a function of $\a$ as depicted in Figure~\ref{fig:canonical-ratio}.

\begin{figure}[ht]
    \begin{center}
        \includegraphics[height=6cm]{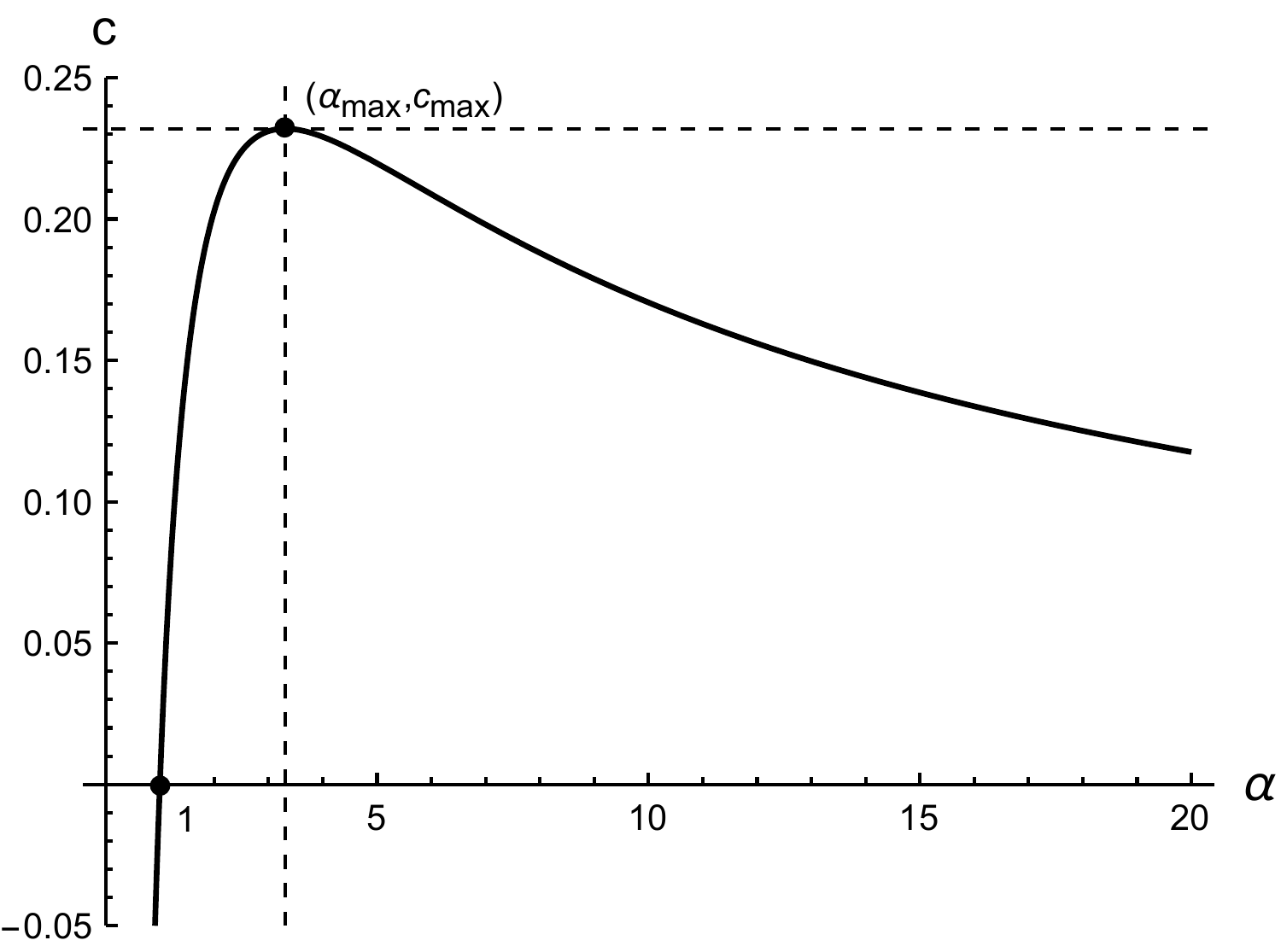}
    \end{center}
    \caption{$c$ as a function of $\a$. Note that $c(1)=0$ and that as $\a$ goes to infinity $c$ goes to zero.}
    \label{fig:canonical-ratio}
\end{figure}

The function $c(\a)$ as in Part (1) of Claim~\ref{clm:reparam} has a unique maximum at the point $(\a_{\max},c_{\max})\approx(3.111,0.232)$, is strictly increasing for $0<\a<\a_{\max}$, and strictly decreasing for $\a>\a_{\max}$. Additionally, we have that $c>0$ for $\a \geq 1$, and as $\a$ goes to infinity, $c$ goes to $0$. Define
	$$ \r_{\max}(n) = \r(\a_{\max}) = \a_{\max}(n+1)-2\ln\left(\frac{1+\a_{\max}}{2}\right)$$
By Observation~\ref{obs:monotone-lambda} and Part (2) of Claim~\ref{clm:reparam}, we have that $M$ increases with $\r$ for $0<\r<\r_{\max}(n)$, attains its maximum value at $\r=\r_{\max}(n)$, and then decreases with $\r$ for $\r>\r_{\max}(n)$. Moreover, $M\geq H$ for $\r\geq n+1$ ($\a\geq 1$) and as $\r$ goes to infinity, $M/H$ tends to 1.

\subsection{Existence of a Nash Equilibrium (for \texorpdfstring{$n\geq3$}{n>=3})}\label{sec:NE-existence}

So far, we know that for any given $n\geq3$ and $\r>0$, if there exists a Nash equilibrium for the game $\fthotel(n,\r)$, then it is a canonical profile (as depicted in Figure~\ref{fig:NE-shape}). The reason for that is that otherwise, there would be at least one server that has an improving move \emph{within} the interval between its neighbors (as depicted in Figure~\ref{fig:local-improve}). We have not yet considered whether a server has an improving move to anywhere \emph{outside} the interval between its neighbors.

Note first that by Theorem~\ref{thm:NE-form}, if for given $n \geq 3$ and $\r$, there exist no canonical pair $H$ and $M$ (that satisfy conditions \eqref{eq:optimal-H}, \eqref{eq:total-regions} and \eqref{eq:minimal-lambda}), then no Nash equilibrium exists. This is due to the fact that in this case, the peripheral servers can always improve by moving closer to their neighbors, which in turn would improve by equalizing the internal regions, and then again the peripheral servers would move closer, this would continue until the internal regions become too small and then an internal server would improve by moving to the hinterland.

Hence, consider the canonical profile $\Xnr$. By Lemmas~\ref{lem:opt-hinterland} and \ref{lem:opt-internal} we have that in each region there is one point that maximizes the payoff. Moreover, by Claim~\ref{clm:no-colocated} the payoff gained from colocating with another server $j$ can always be exceeded by that of an isolated location within one of $j$'s adjacent regions. Therefore, the canonical profile $\Xnr$ is a Nash equilibrium if and only if relocating to the optimal location of any region is not an improving move for any player $i$.

\begin{figure}[t]
    \begin{subfigure}[t]{\textwidth}
        \includegraphics[width=\textwidth]{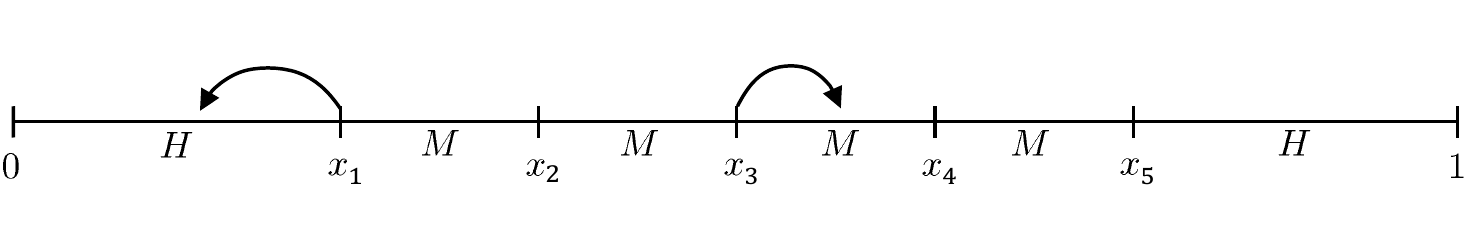}
        \caption{players relocating within the interval between their neighbors.}
        \label{fig:local-improve}
    \end{subfigure}
    \begin{subfigure}[t]{\textwidth}
        \includegraphics[width=\textwidth]{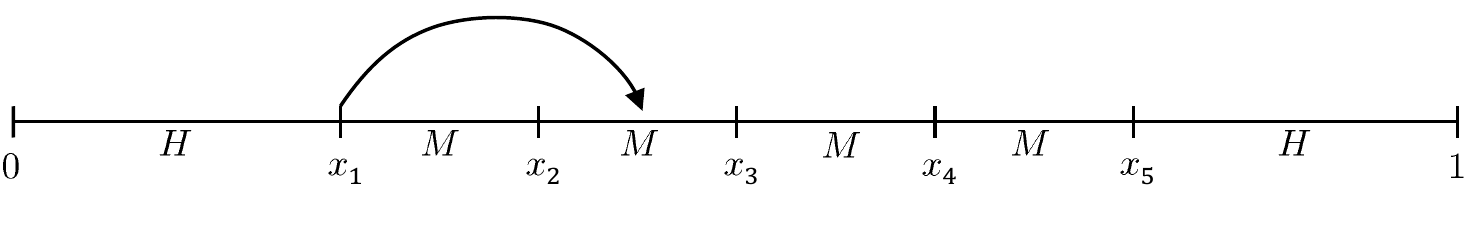}
        \caption{a peripheral player moving to an internal segment.}
        \label{fig:peripheral-improve}
    \end{subfigure}
    \begin{subfigure}[t]{\textwidth}
        \includegraphics[width=\textwidth]{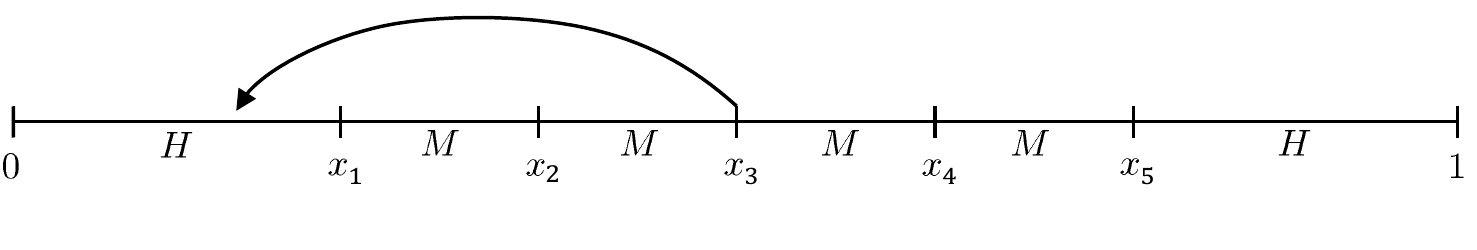}
        \caption{an internal player moving to the hinterland.}
        \label{fig:internal-improve}
    \end{subfigure}
\label{fig:improving-moves}
\caption{possible improving moves.}
\end{figure}

It is easy to see that no internal player has an improving move in an internal region and that no peripheral player has an improving move in the other hinterland (due to the monotonicity of $\EH{d}$ and $\EM{d}$). Furthermore, by Lemma~\ref{lem:region-difference}, the payoff of a peripheral player is larger than that of an internal player, so no peripheral player would gain by moving to an internal region (Figure~\ref{fig:peripheral-improve}). It is left to consider when an internal player would rather move to the hinterland. That is, we must rule out the possibility that the move depicted in Figure~\ref{fig:internal-improve} is an improving move. Formally, this translates to the following.
\begin{observation}\label{obs:NE-conditions}
    Let $\Xnr$ be the canonical profile, with a corresponding canonical pair $H$ and $M$. Then $\Xnr$ is a Nash equilibrium if and only if
    \begin{equation}\label{eq:internal-no-improve}
        \EM{M}+\EM{M} \geq \EH{t^*}+\EM{H-t^*}
    \end{equation}
    where $t^*$ maximizes $\theta(t,H)$ as in Lemma~\ref{lem:opt-hinterland}.
\end{observation}
Notice that the right hand side of \eqref{eq:internal-no-improve} is the payoff gained from locating at the optimal point of a hinterland of length $H$.

This reduces the problem of finding Nash equilibria of the game to finding a canonical pair $H$ and $M$ that also satisfies Equation~\eqref{eq:internal-no-improve}. Hence, from this point forward, we focus on solving Equation~\eqref{eq:internal-no-improve}.

\begin{lemma}\label{lem:M_geq_H}
    If $M\geq H$, then $\Xnr$ is a Nash equilibrium.
\end{lemma}
\begin{proof}
    By Observation~\ref{obs:NE-conditions}, it suffices to show that the canonical pair $H$ and $M$ satisfies Equation~\eqref{eq:internal-no-improve} if $M \geq H$. Consider the right hand side of Equation~\eqref{eq:internal-no-improve}. We have that
        $$ \EH{t^*}+\EM{H-t^*} \leq \EH{t^*}+\EH{H-t^*}
            \leq 2\EH{\frac{H}2}
            \leq 2\EH{\frac{M}2}~.$$
    The first inequality holds since $\EM{t}\leq\EH{t}$ for all $t\geq0$,
    the second follows from Claim~\ref{clm:center-opt}, and the last follows from the assumption and monotonicity.

    Therefore, Equation~\eqref{eq:internal-no-improve} is satisfied if
        $$ \EH{\frac{M}2} \leq \EM{M}~. $$
    Plugging in Claim~\ref{clm:region-profit} we obtain
        $$\frac{1}{\r} \left[1-e^{-\frac{\r M}2}\right] \leq \frac{1}{\r} \left[1-e^{-\r M}\left(1+\frac{\r M}{2}\right)\right]~.$$
    Rearranging we obtain
        $$ e^{\frac{\r M}2} \geq 1+\frac{\r M}{2}~, $$
    which holds for all $\r$ and $M$, concluding the proof.
\end{proof}

\begin{observation}\label{obs:alpha_geq_one}
	$\a\geq 1$ if and only if $M \geq H$.
\end{observation}

\begin{lemma}\label{lem:alpha-beta}
    Fix $n \geq 3$,
    express $\r$ as in Part (4)  
    of Claim~\ref{clm:reparam}. Then $\Xnr$ is a Nash equilibrium if and only if there exists $\b>0$ such that
        \begin{equation}\label{eq:alpha-beta1}
          e^{-\a}(1+\a)=e^{-2\b}(1+\b)~\phantom{.}
        \end{equation}
  and
        \begin{equation}\label{eq:alpha-beta2}
          e^{-\a}\left(1+\frac\a2\right) \leq e^{-\b}\left(\frac34+\frac\b2\right)~.
        \end{equation}
\end{lemma}

\begin{proof}
    Let $s=H-t^*$.
    By Lemma~\ref{lem:region-difference},
    $$ \EH{t^*}+\EM{H-t^*} = \EH{H-s}+\EM{s} = 2\EM{s}+\frac1{2\r}e^{-\r s}~, $$
    so we may write Equation~\eqref{eq:internal-no-improve} as
    $$ \EM{M} \geq \EM{s}+\frac1{4\r}e^{-\r s}~. $$
    Plugging in Claim~\ref{clm:region-profit} and rearranging we obtain
    $$ e^{-\r M}\left(1+\frac{\r M}2\right) \leq e^{-\r s}\left(\frac34+\frac{\r s}2\right)~. $$
    Setting $\a=\r M$ and $\b=\r s$ in the above, we obtain Equation~(\ref{eq:alpha-beta2}). Note that $s$ is uniquely determined by $\a$, and therefore $\b$ may be derived from $\a$.

    By Lemma~\ref{lem:opt-hinterland},
    $$  e^{-\r H} = e^{-2\r s} \left(\frac{1+\r s}2\right)~, $$
    and plugging in Equation~\eqref{eq:optimal-H} we get
    $$ e^{-\r M} \left(\frac{1+\r M}2\right)=e^{-2\r s} \left(\frac{1+\r s}2\right)~,$$
    which translates to Equation~\eqref{eq:alpha-beta1}.

    Since each $\a>0$ defines a unique canonical profile and a unique $\b$, it follows that Equations~\eqref{eq:alpha-beta1} and \eqref{eq:alpha-beta2} are equivalent to Equation~\eqref{eq:internal-no-improve}. Therefore, by Observation~\ref{obs:NE-conditions}, Equations~\eqref{eq:alpha-beta1} and \eqref{eq:alpha-beta2} are necessary and sufficient conditions for a Nash equilibrium.
\end{proof}

We are now ready to prove Theorem~\ref{thm:NE-lower-bound}.

\par\noindent {\em Proof of Theorem~\ref{thm:NE-lower-bound}.~}
By Observation~\ref{obs:alpha_geq_one} and Lemma~\ref{lem:M_geq_H}, if $\a\geq1$ then the canonical profile is a Nash equilibrium. It is left to consider $0<\a<1$. By Lemma~\ref{lem:alpha-beta}, Equations~\eqref{eq:alpha-beta1} and \eqref{eq:alpha-beta2} are sufficient and necessary conditions for a Nash equilibrium. We rewrite these equations as follows.
 \begin{align}
	e^{-\a}(1+\a) &=e^{-2\b_1}(1+\b_1)~, \label{eq:alpha-beta3} \\
	e^{-\a}\left(1+\frac\a2\right) &=e^{-\b_2}\left(\frac34+\frac{\b_2}2\right)~, \label{eq:alpha-beta4} \\
	\b_1 &\leq \b_2~. \label{eq:alpha-beta5}
 \end{align}

Clearly, Equations~\eqref{eq:alpha-beta1} and \eqref{eq:alpha-beta2} hold for $\a$ if and only if Equations~\eqref{eq:alpha-beta3},\eqref{eq:alpha-beta4} and \eqref{eq:alpha-beta5} hold for that $\a$.
Figure~\ref{fig:param-a-b} shows Equations~\eqref{eq:alpha-beta3} and \eqref{eq:alpha-beta4} as parametric curves.
\begin{figure}[ht]
    \begin{center}
        \includegraphics[height=8cm]{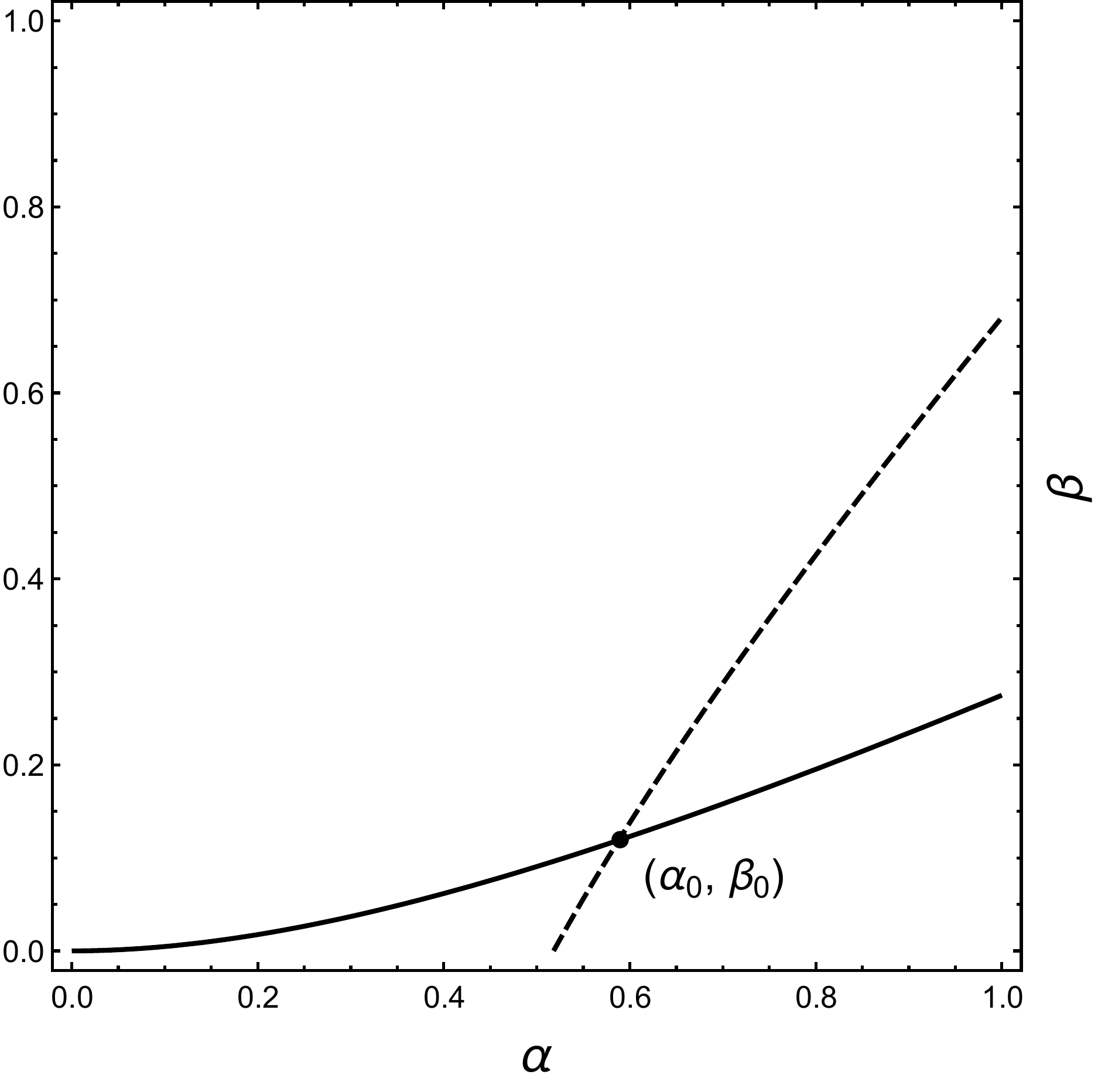}
    \end{center}
    \caption{the solid line is the curve of Equation~\eqref{eq:alpha-beta3} and the dashed line is the curve of Equation~\eqref{eq:alpha-beta4}.}
    \label{fig:param-a-b}
\end{figure}
We next show that for $0<\a<1$ these two curves intersect at a single point $(\a_0,\b_0)$, as depicted in Figure~\ref{fig:param-a-b}. First note that by the implicit function theorem these two curves are continuous and differentiable for $0\leq\a\leq1$. Additionally, it is easy to check that if $\a=0$ then $\b_1<\b_2$ and that if $\a=1$ then $\b_1>\b_2$. Hence, the two curves intersect at least once for $0<\a<1$. To show they intersect \emph{exactly} once, we consider the $\a$ derivatives of $\b_1$ and $\b_2$, and show that, at each intersection point, $d\b_2/d\a>d\b_1/d\a$. Since $\b_1$ and $\b_2$ are continuous and continuously differentiable as functions of $\a$, this suffices to show that they intersect only once.

Accordingly, we now show that $d\b_2/d\a>d\b_1/d\a$ at each intersection point. By implicit differentiation (i.e., taking the $\a$ derivative of both sides of an implicit function) we obtain
$$ -e^{-\a}\,\a  = -2 e^{-2\b_1}(1+2\b_1) \, \frac{d \b_1}{d \a} $$
and
$$ -e^{-\a}\left(\frac{1+\a}{2}\right) = -e^{-\b_2}\left(\frac{1+2\b_2}{4}\right) \frac{d \b_2}{d \a}~.$$
Rearranging, we get
$$ \frac{d \b_1}{d \a} = \frac12 \; e^{2\b_1-\a}\,\left(\frac{\a}{1+2\b_1}\right)  $$
and
$$ \frac{d \b_2}{d \a} = e^{\b_2-\a}\left(\frac{2+2\a}{1+2\b_2}\right) ~.$$
Let $(\a_0,\b_0)$ for $0<\a_0<1$ be an intersection of the curves, i.e., when $\a=\a_0$ we have $\b_1=\b_2=\b_0$. Therefore, $d\b_2/d\a>d\b_1/d\a$ if
	$$ e^{\b_2-\a}\left(\frac{2+2\a_0}{1+2\b_0}\right) > \frac12 \; e^{2\b_0-\a_0}\,\left(\frac{\a_0}{1+2\b_0}\right)~, $$
which yields
\begin{equation}\label{eq:beta-bound}
	4\,\left(1+\frac{1}{\a_0}\right) > e^{\b_0}~.
\end{equation}
Assume $\b_0>0$, otherwise Equation~\eqref{eq:beta-bound} holds and we are done. By Equation~\eqref{eq:alpha-beta3}, we get
	$$ e^{-\a_0}(1+\a_0) = e^{-2\b_0}(1+\b_0) \leq e^{-2\b_0}(1+2\b_0)~.$$
Since $ e^{-x}(1+x) $ is monotone decreasing for $x>0$, it follows that $2\b_0 \leq \a_0$. Hence, since $\a_0<1$, we get
	$$ e^{\b_0} \leq e^{\frac{\a_0}{2}} < e^{\frac{1}{2}} < 4\,\left(1+\frac{1}{\a_0}\right)~,$$
and thus Equation~\eqref{eq:beta-bound} holds. This proves that the curves defined by Equations~\eqref{eq:alpha-beta3} and \eqref{eq:alpha-beta4} intersect only once for $\a>0$.

It follows that Equations~\eqref{eq:alpha-beta1} and \eqref{eq:alpha-beta2} are satisfied if and only if $\a\geq\a_0$. Hence, by Lemma~\ref{lem:alpha-beta} and Claim~\ref{clm:reparam}, the canonical profile $\Xnr$ for $\r=\a(n+1)-2\ln((1+\a)/2)$ is a Nash equilibrium if and only if $\a\geq\a_0$. Since, by Observation~\ref{obs:monotone-lambda}, $\r$ is strictly increasing as a function of $\a$, the theorem follows.

The second part of the theorem is obtained by numerical approximation of the point of intersection $(\a_0,\b_0)$.
\qed

\par\noindent {\em Proof of Corollary~\ref{cor:NE-profiles}.~}
    The maximum possible value for $c$ is $c_{\max}\approx 0.232$, as shown at the end of Section~\ref{sec:canonical-profile}, which yields the lower bound for $H/M$.

    By Theorem~\ref{thm:NE-lower-bound}, a Nash equilibrium exists for $n\geq3$ if and only if $\a \geq \a_0$, and by Part (1) of Claim~\ref{clm:reparam} we have that $1-H/M = c = (\ln((1+\a)/2))/\a$. By simple calculus, the function $c(\a)$ satisfies that $c(\a) \geq c(\a_0)$ if $\a\geq \a_0$. We therefore obtain the upper bound on $H/M$ by setting $c_{\min}=c(\a_0)\approx -0.392$.
\qed

\end{document}